\documentclass[11pt,letterpaper]{article}
  
\usepackage{amsmath,amsthm,nicefrac}
\usepackage{subfigure}
\usepackage{bm}

\usepackage{setspace}
\usepackage[breaklinks]{hyperref}
\hypersetup{colorlinks=true,
            citebordercolor={.6 .6 .6},linkbordercolor={.6 .6 .6},%
citecolor=blue,urlcolor=black,linkcolor=blue}
\usepackage[compact]{titlesec}
\usepackage{authblk}

\usepackage[nameinlink]{cleveref}

\usepackage{epsfig}
\usepackage{amsthm,amssymb,mathrsfs}
\usepackage{xspace}
\usepackage{soul}
\usepackage{latexsym}
\usepackage{bbm}
\usepackage{enumitem}

\usepackage[dvipsnames]{xcolor}
\definecolor{DarkGray}{rgb}{0.66, 0.66, 0.66}
\definecolor{DarkPowderBlue}{rgb}{0.0, 0.2, 0.6}
\definecolor{fluorescentyellow}{rgb}{0.8, 1.0, 0.0}

\usepackage{thmtools,thm-restate}

\usepackage[letterpaper,top=2cm,bottom=2cm,left=3cm,right=3cm,marginparwidth=1.75cm]{geometry}
\makeatletter
\allowdisplaybreaks

\newcounter{note}[section]

\newcommand{\initOneLiners}{%
    \setlength{\itemsep}{0pt}
    \setlength{\parsep }{0pt}
    \setlength{\topsep }{0pt}
}

\usepackage{natbib}
 \bibpunct[, ]{(}{)}{,}{a}{}{,}%

\newtheorem{theorem}{Theorem}
\newtheorem{lemma}{Lemma}

\newtheorem{corollary}{Corollary}

\newtheorem{assumption}{Assumption}

\newtheorem{remark}{Remark}

\newtheorem{definition}{Definition}

\newcommand{\ve}[1]{\bm{#1}}

\newcommand{\R}{\mathbb{R}}
\newcommand{\E}{\mathbb{E}}
\newcommand{\Prob}{\mathbb{P}}
\newcommand{\indi}{\mathbbm{1}}
\newcommand{\indibrac}[1]{\mathbbm{1}_{\{#1\}}}
\newcommand{\abs}[1]{\lvert#1\rvert}
\newcommand{\sumall}{\sum_{i=1}^\njt}
\newcommand{\Thebrac}[1]{\Theta\left(#1\right)}
\newcommand{\obrac}[1]{o\left(#1\right)}
\newcommand{\Obrac}[1]{O\left(#1\right)}
\newcommand{\wbrac}[1]{\omega\left(#1\right)}
\newcommand{\Wbrac}[1]{\Omega\left(#1\right)}

\newcommand{\FCFS}{\textnormal{FCFS}}
\newcommand{\PP}{\textnormal{SNF}}
\newcommand{\sysinf}[1]{#1^{(\infty)}}
\newcommand{\yc}{\Phi}
\newcommand{\ycinf}{\sysinf{\Phi}}
\newcommand{\offset}{\bar{r}}
\newcommand{\cmax}{c_{\max}}
\newcommand{\vmx}{v_{\max}}
\newcommand{\fmx}{f_{\max}}

\newcommand{\njt}{I}
\newcommand{\arr}{\lambda}
\newcommand{\sev}{\mu} 
\newcommand{\siz}{\ell}
\newcommand{\load}{\rho}
\newcommand{\slk}{\delta}
\newcommand{\Lm}{\ell_{\max}}
\newcommand{\sysvar}{\sigma^2}
\newcommand{\mwt}{\E\big[T^w(\infty)\big]}
\newcommand{\wait}{T^w}
\newcommand{\wst}{\delta'} 
\newcommand{\normwork}{\overline{W}}
\newcommand{\fnw}{w}
\newcommand{\fnormwork}[1]{\fnw(#1)}
\newcommand{\Lmratio}{\epsilon_0}
\newcommand{\indexone}{i^*_1}
\newcommand{\indextwo}{i^*_2}
\newcommand{\indexmain}{i^*} 
\newcommand{\sevmin}{\sev_{\min}}
\newcommand{\sevmax}{\sev_{\max}}
\newcommand{\xnom}{\bar{x}}
\newcommand{\sysup}[1]{#1^{(U)}}
\newcommand{\syslow}[1]{#1^{(L)}}

\graphicspath{{figs/}}

\begin{document}

\title{Sharp Waiting-Time Bounds for Multiserver Jobs}

\author[1]{Yige Hong\thanks{\url{yigeh@andrew.cmu.edu}}}
\author[1]{Weina Wang\thanks{\url{weinaw@cs.cmu.edu}}}
\affil[1]{Computer Science Department, Carnegie Mellon University}

\date{}
\maketitle

\begin{abstract}
Multiserver jobs, which are jobs that occupy multiple servers simultaneously during service, are prevalent in today's computing clusters. But little is known about the delay performance of systems with multiserver jobs. We consider queueing models for multiserver jobs in \emph{scaling regimes} where the system load becomes heavy and meanwhile the total number of servers in the system and the number of servers that a job needs become large. 
Prior work has derived upper bounds on the queueing probability in this scaling regime. However, without proper lower bounds, the existing results cannot be used to differentiate between policies. In this paper, we study the delay performance by establishing \emph{sharp bounds} on the \emph{mean waiting time} of multiserver jobs, where the waiting time of a job is the time spent in queueing rather than in service. 
We first characterize the \emph{exact order} of the mean waiting time under the First-Come-First-Serve (FCFS) policy. Then we prove a lower bound on the mean waiting time of all policies, which has an \textit{order gap} with the mean waiting time under FCFS. Finally, we show that the lower bound is \emph{achievable} under a priority policy that we call Smallest-Need-First (SNF).
\end{abstract}


\section{Introduction}\label{sec:introduction}
In today's large-scale computing clusters behind cloud platforms, \textit{multiserver jobs} have become increasingly prevalent, where a multiserver job is a job that demands to occupy multiple ``servers'' (which can be multiple physical servers, multiple CPU cores, etc.) simultaneously during its runtime \citep[][]{TirBarDen_20,VerPedKor_15,LinPaoCho_18,AbaBarChe_16}. 
For example, cloud platforms allow users to specify the number of CPU cores in their virtual machines or containers, and this information can be utilized by centralized schedulers to make scheduling decisions (see, e.g., \citet{VerPedKor_15}, Google Kubernetes Engine \citep{GKE}). Moreover, the number of  ``servers'' that a multiserver job requests, which we refer to as the \textit{server need}, is becoming increasingly large. This trend is driven by machine learning jobs from applications like TensorFlow in \citet{AbaBarChe_16}, where the jobs are highly parallel and require synchronization. According to the statistics from Google's Borg Scheduler in \citet{VerPedKor_15}, the server needs in Borg can vary across six orders of magnitudes.

In this paper, we study the impact of multiserver jobs on the delay performance of large-scale computing systems using queueing models.
Queueing models with multiserver jobs have been studied in the literature, but quantifying the delay performance is notoriously hard.
Exact steady-state distributions can only be derived in highly simplified settings with two servers \citep[][]{BriGre_84,FilKar_06},
while the majority of prior work has focused on characterizing stability conditions \citep[][]{GroHarSch_20,AfaBasGri_19,MorRum_16,RumMor_17}. However, even for stability, exact conditions are known only for the special cases where all jobs have the same service rate or where there are two job classes.
We comment that concurrent to the conference version of our work \citep{HonWan_22}, 
\citet{GroHarSch_22_wfcs} and \citet{GroScuHarSch_22_msj_srpt} study the delay performance of multiserver jobs in the traditional heavy-traffic regime. 
A more detailed review of related work is provided in Section~\ref{sec:related-work}. 

A recent advance in understanding the delay of multiserver jobs is a characterization of the \emph{queueing probability} in a \emph{large system} by \citet{WanXieHar_21_2}, where the queueing probability is the probability that an arriving job has to queue rather than entering service immediately.
Specifically, \citet{WanXieHar_21_2} consider a multiserver job system with $n$ servers, and study the asymptotic scaling regimes where $n$ becomes large.
The scaling regimes allow different job types to have different arrival rates, server needs and service rates. Among those parameters, server needs and arrival rates can scale up with $n$. Such scaling regimes capture the trend that different multiserver jobs can be highly heterogeneous, especially in terms of server needs. They establish an upper bound on the queueing probability, based on which they give a sufficient condition for the queueing probability to diminish as $n$ goes to infinity.

Although the work of \citet{WanXieHar_21_2} identifies when the queueing probability diminishes in large systems, which is a much desirable operating scenario, it does not provide much insight for differentiating between scheduling policies.
In particular, their queueing probability upper bound holds for any scheduling policy that is reasonably work-conserving (although the bound is presented only for the First-Come-First-Serve policy).
Moreover, queueing probability does not directly translate to delay of jobs.

In this paper, we focus on the \emph{waiting time} of jobs, which is the total time a job spends waiting in the queue (not receiving any service), under various scheduling policies.
The waiting time is a performance metric that is directly related to job delay.
Our goal is to establish bounds on the mean waiting time that are \emph{order-wise tight} as the number of servers, $n$, scales.
Such tight bounds will enable us to differentiate between policies based on their delay performance.
We comment that there has been a line of work in the literature \citep[][]{Liu_19,LiuYin_19_2,LiuYin_20,LiuGonYin_20_2,vanZubBor_20,WenWan_21_2,WenZhoSri_21_2} that focuses on quantifying when the mean waiting time diminishes in large systems for various queueing models.
However, little is known on \emph{how fast} the mean waiting time diminishes due to the lack of lower bounds.
Our results provide the rate of diminishing when the mean waiting time does diminish, but our tight bounds on the mean waiting time are not limited to the ``diminishing'' scenario.

Since the First-Come-First-Serve (FCFS) policy is widely used as a default policy in practice and also receives the most attention from theoretical studies of multiserver jobs \citep[][]{BriGre_84,FilKar_06,GroHarSch_20,AfaBasGri_19,MorRum_16,RumMor_17}, in this paper, we will first examine FCFS and understand the exact order of the mean waiting time under it.
Then a natural question that arises is: \emph{can any policy outperform FCFS in terms of the mean waiting time?}
More generally, we aim to answer the following fundamental questions:
\begin{itemize}
\item \emph{What is the optimal order of the mean waiting time as the system scales?
\item Which policy achieves the optimal order?}
\end{itemize}

\subsection{Model and performance metric}
We consider a system that consists of $n$ servers and $\njt$ types of jobs.
An example is illustrated in Figure~\ref{fig:msj}.
Suppose type~$i$ jobs need the simultaneous service of $\siz_i$ servers. We sort the job types such that their \emph{server needs}~$\siz_i$'s satisfy $\siz_1 \leq \siz_2 \leq \dots \le \siz_\njt$. Let the \textit{maximal server need} $\Lm$ be $\Lm = \max_{i\in\{1,2,\dots,\njt\}} \siz_i=\siz_\njt$, and we call type~$\njt$ jobs the \textit{maximal-need jobs}. 

\begin{figure}
    \centering    \includegraphics[width=10cm]{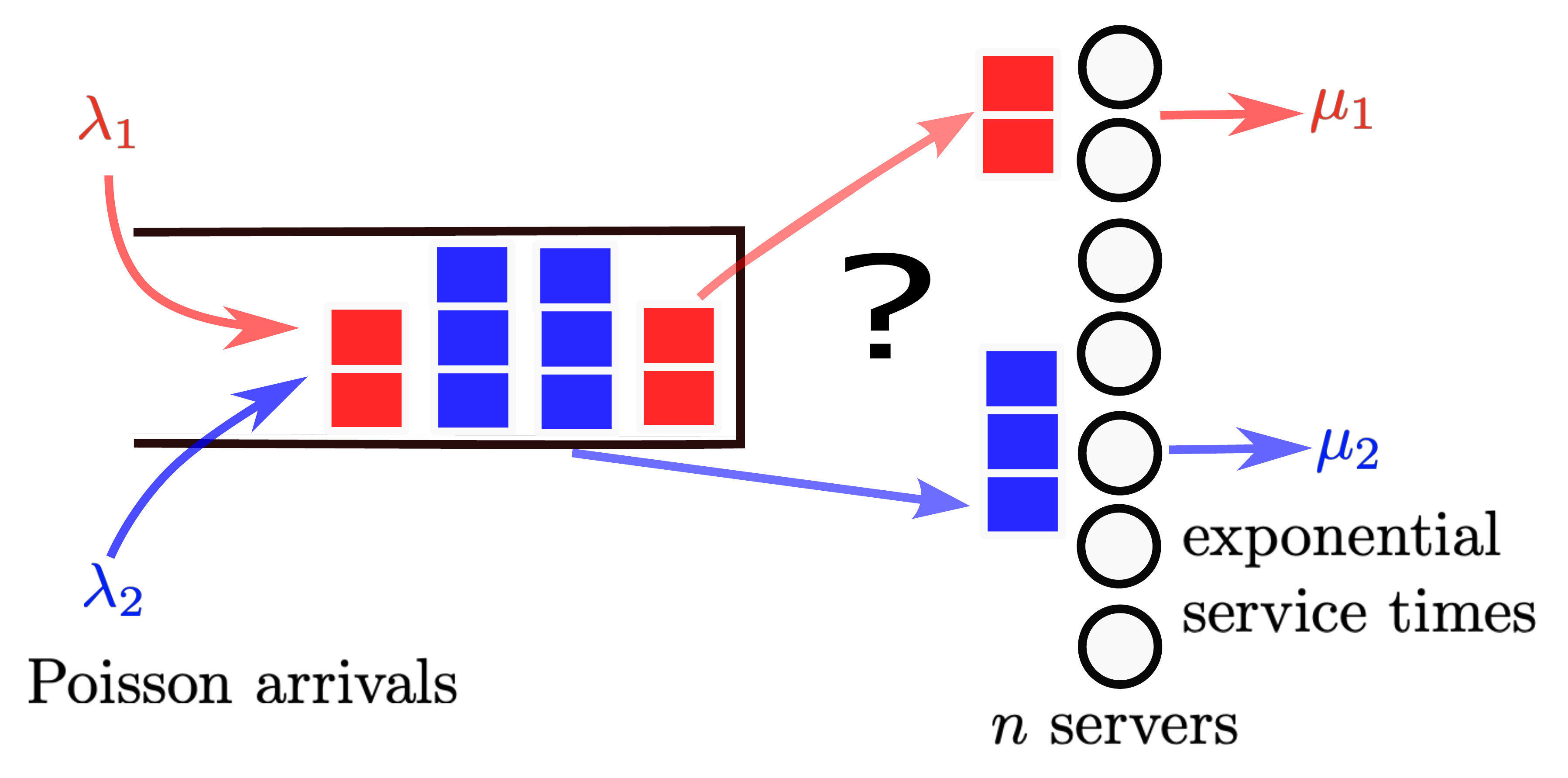}
    \caption{A multiserver-job system with two types of jobs. Type $1$ jobs have arrival rate $\arr_1$, service rate $\sev_1$, and server need $\siz_1 = 2$. Type $2$ jobs have arrival rate $\arr_2$, service rate $\sev_2$, and server need $\siz_2 =3$.}
    \label{fig:msj}
\end{figure}

The dynamics of the system are as follows. For each $i=1,2,\dots, \njt$, type~$i$ jobs arrive to the system following a Poisson process with \emph{arrival rate}~$\arr_i$. Upon arrival, a job either starts service immediately or waits in a centralized queue. When a type~$i$ job starts service, it leaves the queue and makes exclusive use of $\siz_i$ servers. The job leaves the system after receiving enough service. The service time of a type~$i$ job follows an exponential distribution with \emph{service rate}~$\sev_i$. The service times and arrival events are independent. 
During the operation of the system, a scheduling policy is used to determine which set of jobs to serve at any time. The scheduling policy is allowed to be preemptive, i.e., we can put a job in service back to the queue and resume its service later.

We measure the performance of our scheduling policy based on \textit{mean waiting time} as defined below: let $\wait_i(\infty)$ denote the waiting time of type~$i$ jobs in steady-state, then the mean waiting time is defined as the steady-state expected waiting time averaged over all job types, i.e.,
\begin{equation*}
\mwt=\frac{1}{\arr}\sum_{i=1}^\njt \arr_i \E\left[\wait_i(\infty)\right],
\end{equation*}
where $\arr \triangleq \sumall \arr_i$ is the total arrival rate.

\subsection{Scaling regimes}\label{subsec:intro-scaling}
We study job delay in scaling regimes where the number of servers, $n$, goes to infinity.
Specifically, we consider a sequence of systems with parameters scaling up jointly with $n$, and analyze the growth/decrease rate of the mean waiting time.
In the considered scaling regimes, the arrival rates $\arr_i$ and server needs $\siz_i$ are allowed to scale with $n$, while the service rate $\sev_i$ and the number of job types $\njt$ stay constant. One key parameter for specifying a scaling regime is the \textit{slack capacity} $\slk$, defined as $\slk= n - \sum_{i=1}^\njt \frac{\arr_i \siz_i}{\sev_i}$, which is the expected number of idle servers in steady state. Slack capacity is used to specify the heaviness of traffic, which is alternatively specified by \emph{load} $\load$ given by $\load=\sum_{i=1}^\njt \frac{\arr_i\siz_i}{n\sev_i}$ in literature. 

\begin{figure}
    \centering
    \includegraphics[width=8cm]{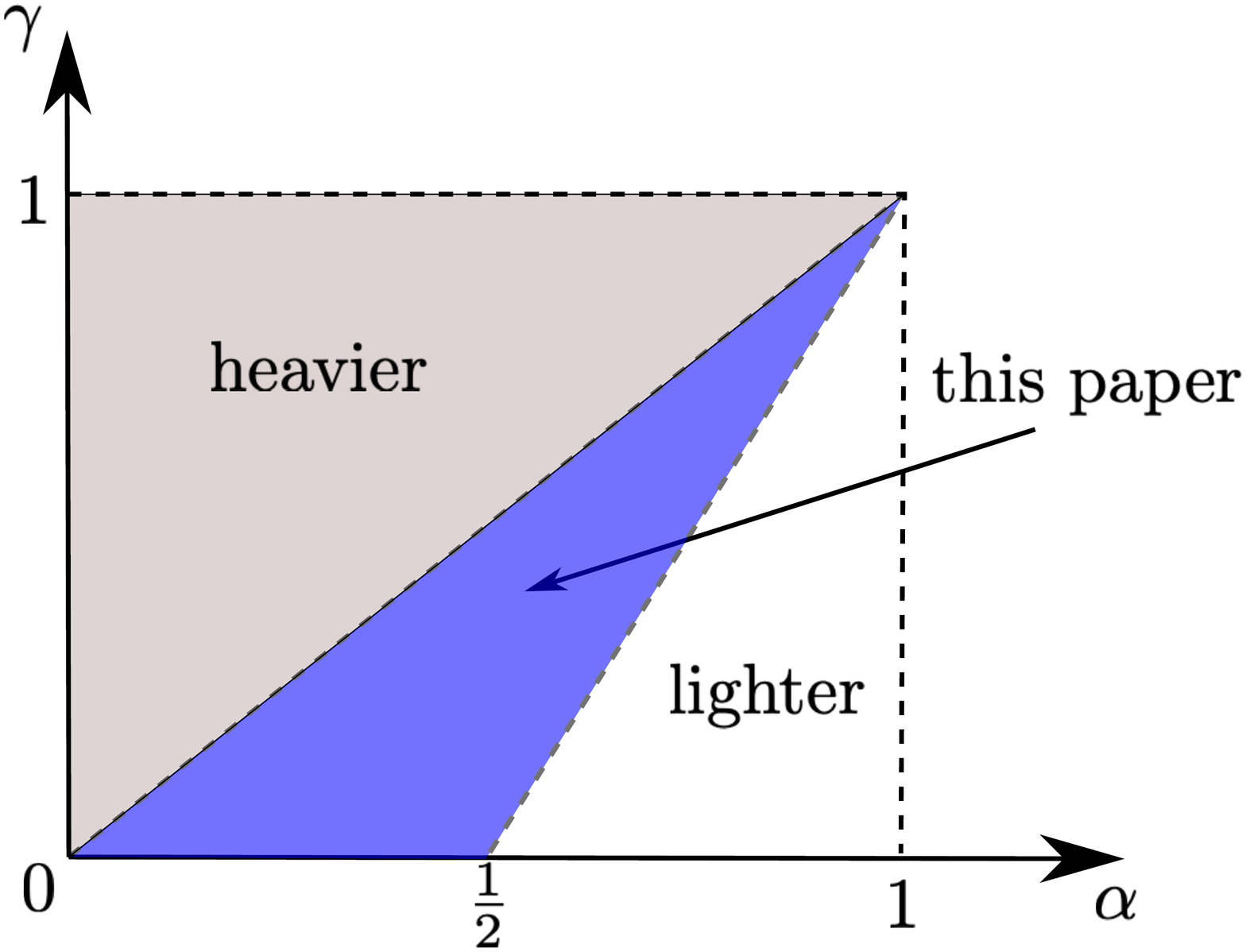}
    \caption{Scaling regimes under the special parameterization with slack capacity $\slk = n^\alpha$ and maximum server need $\Lm = n^\gamma$. The traffic is heavier as we move to the upper left. The three triangles are partitioned by lines $\alpha=\gamma$ and $\alpha=\frac{1+\gamma}{2}$.} 
    \label{fig:traffic}
\end{figure}

For expositional purposes, we now parameterize the scaling regimes in a specific way below to demonstrate our results.
Our general model is presented in Section~\ref{sec:model}.
Suppose $\Lm = n^\gamma$, $\slk = n^\alpha$ for some exponents $0\leq \alpha, \gamma < 1$, and the total arrival rate $\arr = \Theta(n)$.
 We aggregate all scaling regimes with the same $(\alpha, \gamma)$ pair into one point and plot all such points, as shown in Figure~\ref{fig:traffic}. We partition the set of exponent pairs $(\alpha, \gamma)\in [0,1)^2$ into three triangles, using the lines $\alpha=\gamma$ and $\alpha = \frac{1+\gamma}{2}$. The corresponding scaling regimes to the upper left are in general ``heavier'' than the scaling regimes to the lower right, since the former regimes have larger work variability and smaller slack capacity. We comment that the point $(\alpha, \gamma) = (\frac{1}{2}, 0)$ is analogous to the celebrated Halfin-Whitt regime introduced in \citet{HalWhi_81} and $(\alpha, \gamma) = (0,0)$ is analogous to the Non-Degenerate Slowdown (NDS) regime in \citet{Ata_12} in traditional multiclass M/M/$n$ models.

We focus on the scaling regimes where $\gamma < \alpha < \frac{1+\gamma}{2}$, marked in blue in Figure~\ref{fig:traffic}.
The scaling regimes satisfying the condition are not too light; the lighter regimes marked in white, studied in \citet{WanXieHar_21_2}, can be shown to have both queueing probability and mean waiting time diminish at a rate faster than any polynomial in $n$ under any reasonably work-conserving policy.
Meanwhile, the regimes under consideration are not too heavy either; the system still enjoys diminishing mean waiting time together with high system utilization.

\subsection{Results}\label{sec:intro-results}
We present our main results here in the specialized scaling regimes for expositional purposes.
The general forms with fully specified assumptions are presented in Sections~\ref{sec:model} and \ref{sec:main-results}. Our results and analysis heavily use the asymptotic notation. 
\footnote{We use the standard Bachmann–Landau notation. Consider two sequences $a(n)$ and $b(n)$ (or simply $a$ and $b$), where $b(n)$ is positive for large enough $n$. Then $a = \Obrac{b}$ if $\limsup_{n\to\infty}\frac{|a|}{b}<\infty$; $a = \obrac{b}$ if $\lim_{n\to\infty} \frac{a}{b} = 0$; $a = \Wbrac{b}$ if $\liminf_{n\to\infty}\frac{a}{b}>0$, which is equivalent to $b = \Obrac{a}$; $a = \wbrac{b}$ if $\lim_{n\to\infty} \frac{\abs{a}}{b} = \infty$, which is equivalent to $b = \obrac{a}$; $a = \Thebrac{b}$ if $a$ satisfies both $a = \Obrac{b}$ and $a = \Wbrac{b}$.} 

\begin{itemize}
    \item\textbf{Mean waiting time under FCFS.} The exact order of the mean waiting time under FCFS is given by
    \begin{equation}
        \mwt^{\textnormal{FCFS}} = \Thebrac{n^{\gamma-\alpha}}.
    \end{equation}
    \item\textbf{Mean waiting time lower bound.} Under any policy, the mean waiting time is lower bounded as
    \begin{equation}\label{eq:lower-intro}
        \mwt =  \Wbrac{n^{-\alpha}},
    \end{equation}
    \item\textbf{Order-wise optimal policy.} Consider a static priority policy that we call the \emph{Smallest-Need-First (SNF) policy}, which preemptively prioritizes the jobs with smaller server needs. Then the mean waiting time under SNF achieves the lower bound in \eqref{eq:lower-intro}, i.e.,
    \begin{equation}\label{eq:pp-intro}
        \mwt^{\PP} = \Thebrac{n^{-\alpha}}.
    \end{equation}
    Therefore, the SNF policy is order-wise optimal in the mean waiting time.
\end{itemize}

Comparing the mean waiting times under FCFS and under SNF, we can see that FCFS is strictly suboptimal, and SNF improves upon FCFS by a factor of $\Thebrac{n^{\gamma}}$.

A key to proving the mean waiting time results above is the \emph{order-wise tight} bounds on the expected workload we establish (Lemma~\ref{result:workload-lower} and Lemma~\ref{result:workload-upper}).
In addition, although we consider the system under the traffic regime where $\gamma < \alpha < \frac{1+\gamma}{2}$, we still need to analyze ``subsystems'' that are in the lighter traffic regime.
In this lighter regime, we show that the total server need decays faster than any polynomial (Lemma~\ref{lem:proof:total-server-need-expectation}).
All these lemmas hold under a very general class of policies, so they could be potentially relevant when we study policies other than FCFS and SNF.

\subsubsection*{\textbf{Results on queueing probability.}}
As a by-product to our analysis, we further derive an upper bound on the queueing probability under any work-conserving policy, presented in Corollary~\ref{result:queueing-probability}, which significantly improves upon the queueing probability upper bound in \citet{WanXieHar_21_2} in a slightly more constrained traffic regime.

\subsubsection*{\textbf{Simulation experiments.}}
The SNF policy we consider in our analysis is a \emph{preemptive} priority policy.
Preemption is sometimes undesirable in practice.
Therefore, we use simulation experiments to also explore a \emph{non-preemptive} version of the priority policy, which we call the non-preemptive Smallest-Need-First (SNF-NP) policy.
SNF-NP serves a job with the smallest server need in the queue when enough number of servers free up.
Our simulation experiments compare the mean waiting times under FCFS, SNF, and SNF-NP.
The simulation results, presented in Section~\ref{sec:simulation}, show that SNF-NP has comparable performance with SNF, and 
demonstrate
the performance gap between FCFS and SNF/SNF-NP.

\subsection{Technical challenges}
The main technical challenges in analyzing the considered multi-server-job system are rooted in the \emph{heterogeneity} among job types in both their service rates and their server needs.
Such heterogeneity makes the system dynamics multidimensional: neither the total number of jobs in service nor the total number of busy servers determines the current job departure rate.
We comment that even for a classical multi-class M/M/$n$ system, where there are multiple job types with different service rates but all job types have a server need of $1$, finding an optimal scheduling policy is known to be a hard problem, and solutions are available mostly in the so-called Halfin-Whitt heavy-traffic regime through the diffusion control problem \citep[][]{AtaManRei_04, HarZee_04, AtaGur_12}.
Compared with the classical multi-class M/M/$n$ system, our multiserver-job system has an additional layer of intricacy due to the heterogeneous server needs, which makes it possible for the system to have servers idling while there are jobs waiting in the queue.

To address the challenges due to heterogeneity, our analysis relies on various state-space concentration results.
State-space concentration is a phenomenon where the state concentrates around a \emph{subset} of the state space in steady state, observed in queueing systems in heavy-traffic or large-system regimes \citep[][]{WanMagSri_17_2, Liu_19, LiuYin_19_2, LiuYin_20, LiuGonYin_20_2, WenZhoSri_21_2}. 
In the multiserver-job system we consider, state-space concentration results are crucial for analyzing the system dynamics when the queue is nonempty.
The scenario when the queue is nonempty is especially important to our scaling regimes since the queueing probability may not be diminishing even when the mean waiting time is diminishing.
This contrasts with the analysis in the prior work of \citet{WanXieHar_21_2}, which focuses on diminishing queueing probability.
Furthermore, our performance goal is to achieve the \emph{optimal order} of the mean waiting time in large systems, which deviates from the traditional performance goal of minimizing delay or certain long-run cost.

\subsection{\textbf{Organization of the paper and relationship with the conference version}.}
This paper is organized as follows. In Section~\ref{sec:related-work}, we review some additional related work that has not been discussed in the introduction. 
We present our model and assumptions in Section~\ref{sec:model} and then formally state our three main theorems in Section~\ref{sec:main-results}. In Section~\ref{sec:drift-prelim}, we give an overview of the proof structure and preliminaries of our main proof technique, the drift method. In Section~\ref{sec:workload-bounds}, we state and prove Lemmas~\ref{result:workload-lower} and \ref{result:workload-upper}, which will be used to prove the three theorems in Sections~\ref{sec:proof:fcfs}, \ref{sec:proof:lower} and \ref{sec:proof:priority}. Finally, in Section~\ref{sec:simulation}, we present the simulation results.

This paper has the following differences from our previous conference version \citep{HonWan_22}. 
First, we have included a more comprehensive related work section (Section~\ref{sec:related-work}).
Second, we have included proofs of some important lemmas and theorem that were omitted due to the space limit in the conference paper.
These are the proofs of Lemmas~\ref{result:workload-lower} and \ref{result:workload-upper} and Theorem~\ref{result:waiting-time-fcfs}, which can be found in Section~\ref{sec:workload-bounds} and Section~\ref{sec:proof:fcfs}. 
Third, we have changed the name of the order-wise optimal policy that we propose from ``P-Priority'' to ``Smallest-Need-First'' to better reflect the feature of the policy.

\section{Related work}\label{sec:related-work}
In this section, we give a more detailed review of the prior work on multiserver-job models as well as some related models that are not covered in the introduction. 

\subsubsection*{\textbf{Multiserver-job model.}}
As mentioned in the introduction, the majority of prior work on the multiserver-job model has either focused on characterizing stability conditions \citep{GroHarSch_20,AfaBasGri_19,MorRum_16,RumMor_17}, or been restricted to the highly specialized settings with two servers \citep{BriGre_84,FilKar_06}. 
However, recently, there are two papers that study the delay performance of multiserver jobs \citep{GroHarSch_22_wfcs,GroScuHarSch_22_msj_srpt}, which are concurrent to the conference version of our work \citep{HonWan_22}.
\citet{GroHarSch_22_wfcs} characterizes the mean response time in a multiserver-job model under two proposed policies.
One of their policies is called ServerFilling.  \citet{GroScuHarSch_22_msj_srpt} then proposes and analyzes a variant of ServerFilling called ServerFilling-SRPT. 
The biggest distinction between their work and our work is in the scaling regime: in their work, the analysis of mean response time is asymptotically tight when the load of the system approaches one and the number of servers remains \textit{fixed}; in contrast, we consider the scaling regimes where the load, number of servers and server needs \textit{scale jointly}. 
Another distinction is the distributional assumptions on the server needs: their work assumes that the server needs are numbers that can divide the total number of servers, while our work assumes that the maximal server need is small compared with the slack capacity.

\subsubsection*{\textbf{Virtual machine (VM) scheduling.}}
A problem related to the multiserver-job scheduling problem studied in this paper is the virtual machine (VM) scheduling problem (see, e.g., \cite{MagSriYin_12,MagSri_13,MagSriYin_14,XieDonLu_15,PsyGha_18,PsyGha_19,StoZho_21}).
For the VM scheduling problem, typically the system consists of multiple servers, where each server has certain units of each type of resource (e.g., CPU, memory, storage).
A VM job demands to occupy multiple units of each type of resource.
Each VM job will be served on a single server.
Some results for the VM scheduling problem in the traditional heavy-traffic regime can be specialized to the multiserver job scheduling problem.
To see this, consider a VM scheduling problem where the system consists of a single server and there is a single resource type.
Then each unit of resource can be viewed as a server in the multiserver-job scheduling problem.
With this specialization, the results in \citep{MagSriYin_14} provide bounds on a linear combination of the queue lengths of different types of jobs.
The bounds are tight in the traditional heavy-traffic regime with a fixed amount of resources (number of servers in the multiserver-job setting).
However, these bounds do not directly translate into heavy-traffic optimality of mean job response time.

\subsubsection*{\textbf{Multitask job model.}}
A multitask job is a job that consists of a batch of tasks that can run on servers in parallel, which is similar to a multiserver job in that both can occupy multiple servers at the same time. However, unlike a multiserver job, the tasks of a multitask job can have different runtimes and do not need to be executed simultaneously. Multitask job model has been considered under a wide variety of settings, and is sometimes referred to as batch arrival model \citep[see, e.g.,][]{Mil_59, DawPen_19, DawFraPen_20, DawHamPen_19}. Recently, multitask job model is also extensively studied under the setting of parallel computing due to the popularity of large-scale data processing systems such as MapReduce, Apache Hadoop and Apache Spark. \citep[see, e.g.,][]{WenWan_21_2, Zub_20}. The work closest to our work is \citet{WenWan_21_2}, which shows diminishing queueing time for multi-server jobs in a load-balancing system where tasks of a job need to be dispatched to the queues at the servers upon arrival.

\subsubsection*{\textbf{Dropping model.}}
When the multiserver-job system does not have any queueing space and allows incoming jobs to be dropped, it becomes a model that has been studied in the literature and we refer to it as the \emph{dropping model}.
In this model, one can design a dropping policy that decides whether to drop an incoming job or not based on the types of the incoming job and of the jobs currently in service.
Under the policy that drops an incoming job only when it cannot fit into the servers, i.e., when its server need is larger than the number of available servers, the stationary distribution has a product form under exponentially distributed service times, as observed by \citet{ArtKau_79}.
The results have been generalized by \citet{Whi_85} to allow jobs to demand multiple resource types (e.g., both CPU and I/O) and by \citet{van_89} to allow general service time distributions. \citet{Tik_05} further
combined aspects of \citet{Whi_85} and \citet{van_89}.
Different dropping policies, which mostly fall within the class of trunk reservation policies, have been designed to minimize the cost associated with dropping \citep[][]{HunKur_94,BeaGibZac_95,HunLaw_97}.

\subsubsection*{\textbf{Streaming model.}}
The \emph{streaming model} for a communication network resembles the multiserver-job model in many aspects.
In a streaming model, the ``servers'' correspond to the bandwidth in the network and the ``jobs'' are data flows such as audio or video flows.
Then flows that require a fixed amount of bandwidth \citep[][]{Mel_96,DasSri_99,BenFreDel_01,PonKimMel_10}, sometimes referred to as streaming flows, can be viewed as multiserver jobs.
However, a communication network also features a network structure that the multiserver-job model does not have.
A communication network usually has both streaming flows and flows that are flexible in their bandwidth needs, and streaming flows again operate in the dropping model.
The performance metric in such a system typically combines the cost associated with dropping for streaming flows and the cost associated with delay for other flows.

\section{Model}\label{sec:model}
A basic description of the system parameters and dynamics has been given in the introduction section.
In this section, we provide formal descriptions of the scheduling policies, the system states, the scaling regime, and the concept of subsystems used in our analysis.

\subsubsection*{\textbf{Scheduling policies.}}
A scheduling policy decides which jobs to put into service at any moment of time.  
We are interested in the following two policies:
\begin{itemize}
    \item \emph{First-Come-First-Serve (FCFS)}: Jobs are placed onto servers in a First-Come-First-Serve fashion until either the next job in queue does not fit or all the jobs are in service.
    \item \emph{Smallest-Need-First (SNF)}: Recall that the job types are indexed in a way such that $\siz_1\leq \siz_2\leq \dots \leq \siz_\njt$. We assign priorities to job types such that a smaller index has a higher priority. Whenever there is a job arrival or departure, SNF preempts all the jobs in service and determines a new schedule from scratch.  
    SNF starts from job type~$1$ and places as many type~$1$ jobs as possible onto servers.  After this, if there are still servers available, SNF goes to the next priority level, type~$2$, and places as many type~$2$ jobs as possible onto servers.  This procedure continues until no more jobs in the queue can fit into the servers.
\end{itemize}

\subsubsection*{\textbf{System state.}}
Under FCFS or SNF, a Markovian representation of the system state can be described as follows. The state $\ve{u}$ of the Markov chain is an ordered list of the jobs in the system, sorted in their order of arrival, and each entry of $\ve{u}$ describes the type of the corresponding job and whether the job is in service or not. Let the state space be denoted as $\mathcal{U}$. Although the state space is infinite dimensional, in our analysis, we typically only need to focus on three $\njt$-dimensional vectors defined below.

For any time $t$ and each job type~$i$, let $X_i(t)$ denote the number of type~$i$ jobs in the system, $Z_i(t)$ denote the number of type~$i$ jobs in service, and $Q_i(t)\triangleq X_i(t)-Z_i(t)$ denote the number of type~$i$ jobs waiting in the queue. 
Note that since the total number of servers in use cannot exceed $n$, and we cannot serve more jobs than there are in the system, we have the following constraints:
\begin{equation}\label{eq:assump:job-constraints}
    \begin{aligned}
        \sum_{i=1}^\njt \siz_i Z_i(t) &\le n \quad\text{for all }t\geq 0,\\
        Z_i(t) &\leq X_i(t) \quad \text{for all } t\geq 0, i\in[\njt],
    \end{aligned}
\end{equation}
where $[\njt]$ denotes the index set $\{1,2, \dots, \njt\}$.

Let $X_i(\infty)$, $Z_i(\infty)$, and $Q_i(\infty)$ be random variables that follow the corresponding steady-state distributions when they exist.
We sometimes use vector representations of these quantities for convenience.  For example, we write $\ve{X}(t)=\left(X_1(t),X_2(t),\dots,X_\njt(t)\right)$.  We define the vectors $\ve{Z}(t)$, $\ve{Q}(t)$, $\ve{X}(\infty)$, $\ve{Z}(\infty)$, and $\ve{Q}(\infty)$ in a similar way.
Note that these random elements correspond to the $n$ server system and thus their distributions depend on $n$.
Throughout this paper, for conciseness, we often omit the $(\infty)$ in the steady-state random elements except in theorem or lemma statements.

Recall that our performance metric is the mean waiting time $\mwt$, given by
$$\mwt=\frac{1}{\arr}\sum_{i=1}^\njt \arr_i \E\big[\wait_i(\infty)\big],$$
where $\wait_i(\infty)$ is the waiting time of type~$i$ jobs in steady-state.
Note that by Little's law, the mean waiting time can be written as
\begin{equation*}
\mwt=\frac{1}{\arr}\sum_{i=1}^\njt\E\left[Q_i(\infty)\right].
\end{equation*}
Therefore, bounding the mean waiting time reduces to bounding the expected total queue length.

\subsubsection*{\textbf{Scaling regimes.}}
Recall that we consider a scaling regimes where number of servers, $n$, goes to infinity, and the arrival rates $\arr_i$ and server needs $\siz_i$ are allowed to scale with $n$, \textit{while the service rate $\sev_i$ and the number of job types $\njt$ stay constant}.
The scaling regimes are specified by the slack capacity $\slk\triangleq n - \sum_{i=1}^\njt \frac{\arr_i \siz_i}{\sev_i}$, the maximal server need $\Lm \triangleq\max_{i\in[\njt]} \siz_i=\siz_\njt$ and another parameter called the \textit{work variability}: $\sysvar \triangleq \sumall \frac{\arr_i \siz_i^2}{\sev_i^2}$. Work variability reflects the variability of the ``work'' caused by job arrivals in terms of server--time product, which is $\frac{\siz_i}{\sev_i}$ in expectation for each type~$i$ job. To help later presentation, we also define the  \emph{load brought by type~$i$ jobs} $\load_i$ as $\load_i = \frac{\arr_i \siz_i}{n\sev_i}$.

We state our assumptions below. Throughout the paper, $\log n$ denotes natural logarithm.
\begin{assumption}[Heavy traffic assumption]\label{assump:traffic}
     The slack capacity $\slk$ is small relative to $\sqrt{\sysvar}$:
    \begin{equation}\label{eq:assump:traffic}
         \slk=o\left(\frac{\sqrt{\sysvar}}{\log n}\right).
    \end{equation}
\end{assumption}

\begin{assumption}[Maximal server need assumption]\label{assump:lm-bound}
     There exists a constant $\Lmratio$ with $0 < \Lmratio < 1$ such that
    \begin{equation}\label{eq:assump:lm-bound}
        \Lm \leq \Lmratio \slk.
    \end{equation}
\end{assumption}

\begin{assumption}[Commonness assumption]\label{assump:commonness}
     The load brought by the maximal-need jobs is not too small:
    \begin{equation}\label{eq:assump:commonness}
        \rho_{\njt} \triangleq \frac{\arr_{\njt}\siz_{\njt}}{n\sev_{\njt}} = \wbrac{\sqrt{\frac{\slk \log n}{\sqrt{\sysvar} }\cdot \frac{\Lm}{n}}\log n}.
    \end{equation}
\end{assumption}

Assumption~\ref{assump:traffic} guarantees that the traffic is not too light, while Assumption~\ref{assump:lm-bound} guarantees that the system is stable under FCFS and SNF. In the simplified setting of Section~\ref{sec:introduction} where $\Lm=n^\gamma$ and $\slk=n^\alpha$, the first two assumptions correspond to $\alpha < \frac{1+\gamma}{2}$ and $\alpha > \gamma$, which exclude the white and grey parts in Figure~\ref{fig:traffic}, respectively.  Assumption~\ref{assump:commonness} states that the load brought by the maximal-need jobs are not too small. To understand the right hand side expression in Assumption~\ref{assump:commonness}, note that it is automatically satisfied when $\load_{\njt} = \wbrac{\sqrt{\frac{\Lm}{n}}\log n}$. For example, when $\Lm = \Thebrac{\sqrt{n}}$, then it suffices to have $\load_{\njt} = \wbrac{n^{-1/4}\log n}$. However, when the traffic becomes heavier, i.e., when $\frac{\slk \log n}{\sqrt{\sysvar} }$ becomes smaller, Assumption~\ref{assump:commonness} in \eqref{eq:assump:commonness} can be much weaker than $\load_{\njt}  = \wbrac{\sqrt{\frac{\Lm}{n}}\log n}$.

To have an intuitive view of the magnitudes of the parameters, we give the following asymptotics:
$\sysvar = \Obrac{n\Lm}$,  $\slk = \obrac{\frac{n}{(\log n)^2}}$, and $\Lm \leq \Lmratio \slk = \obrac{\frac{n}{(\log n)^2}}$. They can be verified using the definitions and assumptions.

\subsubsection*{\textbf{Subsystems.}}
In our analysis, we frequently use the concept of the $i$-th \emph{subsystem}, which is the system that has all type~$j$ jobs in the original system with $j\leq i$ and removes all type~$k$ jobs with $k\geq i$.
In the $i$-th subsystem, the slack capacity becomes $\slk_i = n - \sum_{j=1}^i \frac{\arr_j \siz_j}{\sev_j}$, and the work variability becomes $\sysvar_i = \sum_{j=1}^i \frac{\arr_j \siz_j^2}{\sev_j^2}$.
Note that $\slk = \slk_{\njt}$ and $\sysvar_{\njt} = \sysvar$.
The maximal server need in the $i$-th system is $\siz_i$ since $\siz_1 \leq \siz_2 \leq \dots \le \siz_i$.

As $i$ increases, the load of the $i$-th subsystem gets heavier since $\slk_i$ becomes smaller. There is a \emph{critical index} $\indexmain$ such that 
\begin{equation}\label{eq:istar}
\indexmain = \min\left\{i\in[\njt]\;\middle|\; \slk_i=\obrac{\frac{\sqrt{\sysvar_i}}{\log n}}\right\},
\end{equation}
i.e., the $\indexmain$th subsystem is the smallest subsystem whose traffic regime is as heavy as that of the original system.
Recall that we have assumed $\slk=\obrac{\frac{\sqrt{\sysvar}}{\log n}}$, so the set in \eqref{eq:istar} contains at least the index $I$ and thus $\indexmain$ is well-defined. 
Note that $\slk_i$ is monotonically decreasing while $\sysvar_i$ is monotonically increasing.  Thus
the index $\indexmain$ serves as a division point: 
for any $i$ with $\indexmain \leq i \leq \njt$, we have $\slk_i = \obrac{\frac{\sqrt{\sysvar_i}}{\log n}}$, resulting in a heavier traffic regime; and for any $i$ with $1 \leq i < \indexmain$, we have $\slk_i = \Wbrac{\frac{\sqrt{\sysvar_i}}{\log n}}$, resulting in a lighter traffic regime.

\section{Main results}\label{sec:main-results}
In this section, we first
present our main results under the scaling regimes we specify in Section~\ref{sec:model} as Theorems~\ref{result:waiting-time-fcfs}, \ref{result:waiting-time-lower-bound}, and \ref{result:waiting-time-priority}.
Then, to demonstrate our results in a more intuitive fashion, we consider the parameterized scaling regimes defined in Section~\ref{subsec:intro-scaling} as a special case, and present the specialized form of our results as Corollary~\ref{coro:parameterized}.

\begin{theorem}[Mean waiting time under FCFS]\label{result:waiting-time-fcfs}
    Consider the multiserver-job system with $n$ servers satisfying Assumptions~\ref{assump:traffic} and \ref{assump:lm-bound}. Under the FCFS policy, for each $i\in [\njt]$, the expected waiting time of type~$i$ jobs satisfies
    \begin{align} 
    \E\big[\wait_i(\infty)\big]^\FCFS & \ge \frac{\sysvar}{n(\slk+\Lm)}\cdot \left(1-o(1)\right),\\
    \E\big[\wait_i(\infty)\big]^\FCFS & \leq \frac{\sysvar}{n(\slk-\Lm)} \cdot \left(1+o(1)\right).
    \end{align}
    Consequently,
    \begin{equation}
    \mwt^\FCFS = \Thebrac{\frac{\sysvar}{n\slk}}, \quad  \E\big[\wait_i(\infty)\big]^\FCFS = \Thebrac{\frac{\sysvar}{n\slk}}.
    \end{equation}
\end{theorem}

\begin{theorem}[Mean waiting time lower bound]\label{result:waiting-time-lower-bound}
    Consider the multiserver-job system with $n$ servers satisfying Assumptions~\ref{assump:traffic} and \ref{assump:lm-bound}. Under any policy, the mean waiting time is lower bounded as 
    \begin{equation}\label{eq:result:mean-waiting-time-lower-bound}
        \mwt  \geq \max_{\indexmain \leq i \leq \njt} \frac{1}{\arr} \frac{\sevmin\sysvar_i}{\siz_i \slk_i} \cdot (1 - o(1)) 
        = \Wbrac{\max_{\indexmain \leq i \leq \njt}\frac{1}{\arr} \frac{\sysvar_i}{\siz_i \slk_i}}
    \end{equation}
    where $\indexmain$ is the critical index defined in \eqref{eq:istar} and $\sevmin=\min_{i\in[\njt]}\sev_i$, and the expression represented by $o(1)$ is independent of the policies.
\end{theorem}

\begin{theorem}[Mean waiting time under SNF]\label{result:waiting-time-priority}
    Consider the multiserver-job system with $n$ servers satisfying Assumptions~\ref{assump:traffic}, \ref{assump:lm-bound}, and \ref{assump:commonness}. Under the SNF policy, the mean waiting time satisfies
    \begin{equation}\label{eq:result:mean-waiting-time-priority}
        \mwt^\PP \leq \frac{1}{\arr} \sum_{i=\indexmain}^\njt \frac{\sevmax \sysvar_i}{\siz_i (\slk_i - \siz_i)} \cdot \left(1+o(1)\right)
        =\Obrac{\max_{\indexmain \leq i \leq \njt}\frac{1}{\arr} \frac{\sysvar_i}{\siz_i \slk_i}},
    \end{equation}
    where $\indexmain$ is the critical index defined in \eqref{eq:istar} and $\sevmax=\max_{i\in[\njt]}\sev_i$.
    Consequently, the SNF policy achieves the optimal order of the mean waiting time.
\end{theorem}

We have a more general bound for the SNF policy that holds without Assumption~\ref{assump:commonness}. Interested readers can refer to Appendix~\ref{sec:discussion}.

Below we state the results appearing in Section~\ref{sec:introduction} as direct consequences to the above theorems.

\begin{corollary}[Mean waiting times in the parameterized scaling regimes]\label{coro:parameterized}
    Consider the multiserver-job system with $n$ servers satisfying Assumptions~1, 2 and 3. Suppose the maximal server need $\Lm = n^\gamma$ and the slack capacity $\slk = n^\alpha$, then the assumptions simplify to $0\leq\gamma < \alpha < \frac{1+\gamma}{2}<1$, $\rho_\njt = \Thebrac{1}$. We further assume that the total arrival rate $\arr = \Theta(n)$. Then we have the following results:
    \begin{enumerate}
        \item Under the FCFS policy, for each $i\in[\njt]$, the expected waiting time of type~$i$ jobs satisfies
        \begin{equation}
            \E\big[\wait_i(\infty)\big]^\FCFS = \Thebrac{n^{\gamma - \alpha}},
        \end{equation}
        and the mean waiting time over all job types also satisfies
        \begin{equation}
            \mwt^\FCFS = \Thebrac{n^{\gamma - \alpha}}.
        \end{equation}
    \item Under any policy, the mean waiting time is lower bounded as
    \begin{equation}
        \mwt = \Wbrac{n^{ - \alpha}},
    \end{equation}
    where the expression represented by $\Wbrac{n^{-\alpha}}$ is independent of the policies.
    \item The mean waiting time under the SNF policy satisfies
    \begin{equation}
        \mwt^\PP = \Thebrac{n^{ - \alpha}}.
    \end{equation}
    \end{enumerate}
\end{corollary}

\section{Proof Roadmap and drift method preliminaries}\label{sec:drift-prelim}
We organize our proofs of the main results as follows: we first prove two important bounds for a quantity called \textit{workload} given by $\sumall \frac{\siz_i}{\sev_i} Q_i$, in Lemma~\ref{result:workload-lower} and Lemma~\ref{result:workload-upper}, respectively. 
Then we convert the workload bounds to the waiting time bounds in Theorem~\ref{result:waiting-time-fcfs} and Theorem~\ref{result:waiting-time-lower-bound} using properties of FCFS and a linear programming relaxation. For Theorem~\ref{result:waiting-time-priority}, we analyze SNF by considering each $i$-th subsystems for $i\in[\njt]$. Some subsystems only need Lemma~\ref{result:workload-lower} and \ref{result:workload-upper}, while others require an additional Lemma~\ref{lem:proof:total-server-need-expectation}.

Our proof approach is closely related to the recently developed drift method \citep[][]{ErySri_12, MagSri_16}. The drift method allows us to extract information from a continuous-time Markov chain $\{\ve{S}(t)\}_{t\geq 0}$ on state-space $\mathcal{S}$ by computing the \textit{drift} of different test functions. Because $S(t)$ is a Markov chain with countable state space and bounded transition rates, we can define \textit{drift} of the function $f$ as
\begin{equation}
    G f(\ve{s}) \triangleq \lim_{t\to 0}\E\left[\frac{f\left(\ve{S}(t)\right) - f(\ve{s})}{t} \;\middle|\; \ve{S}(0)=\ve{s}\right].
\end{equation}
We call the operator $G$ the \textit{generator} of the Markov chain.

For a multiserver-job system, let $\njt$-dimensional real vectors $\ve{x}, \ve{z} \in \R_+^\njt$ be possible realizations of state descriptors $\ve{X}(t)$ and $\ve{Z}(t)$, where recall that $\ve{X}(t)$ is the vector of the number of jobs in the system at time $t$, and $\ve{Z}(t)$ is the vector of the number of jobs in service at time $t$. We focus on $f$ that only depends on $\ve{x}$, i.e., $f: \R_+^\njt \to \R$. The drift of $f$ is of the form
\begin{equation}\label{eq:proof:drift-form}
    G f(\ve{x}, \ve{z}) = \sumall \arr_i \left(f(\ve{x}+ \ve{e}_i) - f(\ve{x})\right)+ \sumall \sev_i z_i \left(f(\ve{x}-\ve{e}_i) - f(\ve{x})\right),
\end{equation}
where $\ve{e}_i \in \R_+^\njt$ is the vector whose $i$-th entry is $1$ and all the other entries are $0$. Note that although $f$ and $G f$ are functions of the system state $\ve{u}$, we write $f(\ve{x})$ and $G f(\ve{x},\ve{z})$ to highlight the variables that affect their values.

We frequently use the following relation regarding the drift
\begin{equation}\label{eq:proof:drift-method-master}
    \E[G f(\ve{X},\ve{Z})] = 0.
\end{equation}
Heuristically, this is because when $\ve{X}(0)$ and $\ve{Z}(0)$ follow the stationary distribution, $\ve{X}(t)$ and $\ve{Z}(t)$ also follow the stationary distribution, so $f(\ve{X}(t))$ and $f(\ve{X}(0))$ have the same expectation. Rigorously speaking, this relation only holds for well behaved functions and Markov processes. The conditions under which the relation holds are discussed in detail in Appendix~\ref{sec:app:drift-cond}. Throughout the paper, we assume \eqref{eq:proof:drift-method-master} holds for all $f$ that we consider.

\section{Workload bounds}\label{sec:workload-bounds}
In this section, we prove two bounds for a quantity called \textit{workload} given by $\sumall \frac{\siz_i}{\sev_i} Q_i$. These bounds are fundamental to the proofs of the main theorems. 
In Lemma~\ref{result:workload-lower}, we give a lower bound on the expected workload applicable to \emph{any policy}.
In Lemma~\ref{result:workload-upper}, we give upper bounds on the expected workload under any \emph{$\wst$-work-conserving policy}, a class of policies defined in Definition~\ref{def:work-conserving}. After stating these two lemmas, we go through preliminaries and proof sketches of the two lemmas. Finally, we show the full proof of the two lemmas at the end of the section.

\begin{lemma}[Workload lower bound]\label{result:workload-lower}
    Consider the multiserver-job system with $n$ servers satisfying Assumptions \ref{assump:traffic} and \ref{assump:lm-bound}. Under any policy, the expected workload is lower bounded as
    \begin{equation}\label{eq:result:workload-lower}
        \E\left[\sumall \frac{\siz_i}{\sev_i} Q_i(\infty)\right] \geq \frac{\sysvar}{\slk} \cdot (1 - o(1)),
    \end{equation}
    where the expression represented by $o(1)$ is independent of the policies.
\end{lemma}

\begin{definition}\label{def:work-conserving}
    We call a policy $\wst$-work-conserving, if the following equation holds
    \begin{equation}
        \sumall \siz_i Z_i(t) \geq \min\left(\sumall \siz_i X_i(t), n - \wst \right) \quad \forall t\geq 0 .
    \end{equation}
\end{definition}
Here $\sumall \siz_i Z_i(t)$ is equal to the number of busy servers at time~$t$, while $\sumall \siz_i X_i(t)$, which we call the \textit{total server need}, is the potential number of busy servers if we can put all jobs at time $t$ into service. Therefore, under a $\wst$-work-conserving policy, either all jobs are in service, or there are at most $\wst$ idling servers. Under any $\Lm$-work-conserving policy, one can show that the system is stable when $\Lm \leq \Lmratio \slk$ (Assumption~\ref{assump:lm-bound}) holds. In particular, the system is stable under both FCFS and SNF. 

\begin{lemma}[Workload upper bound]\label{result:workload-upper}
    Consider the multiserver-job system with $n$ servers under a $\wst$-work-conserving policy with $\wst \leq \Lmratio \slk$, where $ \Lmratio \in (0,1)$ is the parameter in Assumption~\ref{assump:lm-bound}. Then when $\slk = \obrac{\frac{\sqrt{\sysvar}}{\log n}}$, 
    \begin{equation}\label{eq:result:workload-upper-case1}
        \E\left[\sumall \frac{\siz_i}{\sev_i} Q_i(\infty) \right] \leq \frac{\sysvar}{\slk - \wst} \cdot (1 + o(1)) = \Obrac{\frac{\sysvar}{\slk}};
    \end{equation}
    when $\slk = \Wbrac{\frac{\sqrt{\sysvar}}{\log n}}$, 
    \begin{equation}\label{eq:result:workload-upper-case2}
        \E\left[\sumall \frac{\siz_i}{\sev_i} Q_i(\infty) \right] = \Obrac{\sqrt{\sysvar}\log n}.
    \end{equation}
\end{lemma}

\begin{remark}
When Assumption~\ref{assump:traffic} is satisfied, i.e., when $\slk = \obrac{\frac{\sqrt{\sysvar}}{\log n}}$, the workload upper bound in Lemma~\ref{result:workload-upper} coincides with the workload lower bound in Lemma~\ref{result:workload-lower} order-wise, which implies that the expected workload $\E\left[\sumall \frac{\siz_i}{\sev_i} Q_i(\infty) \right]  = \Thebrac{\frac{\sysvar}{\slk}}$. 
Note that in this case, although the expected workload under all $\wst$-work-conserving policies has the same order, the mean waiting time can vary among policies, as shown for FCFS and SNF in Theorems~\ref{result:waiting-time-fcfs} and \ref{result:waiting-time-priority}.
\end{remark}

\subsection{\textbf{Preliminaries and proof sketches for Lemma~\ref{result:workload-lower} and Lemma~\ref{result:workload-upper}.}}\label{sec:workload-bounds-preliminary-sketch}

Our proofs focus on bounding the \textit{normalized work}, defined as
\[
    \normwork \triangleq \sumall \frac{\siz_i}{\sev_i} \left(X_i - \xnom_i\right),
\] 
where we write $\xnom_i \triangleq \frac{\arr_i}{\sev_i}$ for notational simplicity. We claim that normalized work has the same expectation as the workload, i.e., $\E[\normwork] = \E\left[\sumall \frac{\siz_i}{\sev_i} Q_i\right]$. To see this, recall that $Q_i=X_i-Z_i$. Now consider the drift of $X_i$, given by $G X_i=\arr_i-\sev_i Z_i$.  One can verify that $X_i$ satisfies $\E[G X_i]=0$, and thus $\E[Z_i]=\frac{\arr_i}{\sev_i}$.
Therefore, the expected workload can be written as:
\begin{equation}
    \E\left[\sumall \frac{\siz_i}{\sev_i} Q_i\right] = \E\left[\sumall \frac{\siz_i}{\sev_i}\left(X_i - Z_i\right)\right] =\E\left[\sumall \frac{\siz_i}{\sev_i}\left(X_i - \xnom_i\right)\right] = \E[\normwork].
\end{equation}
Therefore, bounding the expected workload is equivalent to bounding the steady-state expectation of the normalized work $\E[\normwork]$. 

We break $\E[\normwork]$ into three terms:
\begin{equation}
    \E[\normwork] = \offset + \E[(\normwork - \offset)^+] - \E[(\normwork - \offset)^-],
\end{equation}
where $\offset \in \R$ is up to our choice; $(\normwork - \offset)^+ \triangleq \max\{\normwork - \offset, 0\}$ denotes the positive part, and $(\normwork - \offset)^- \triangleq -\min \{\normwork - \offset, 0\}$ denotes the negative part. 

The major difficulty during the proofs is bounding the expectation of the positive part $\E[(\normwork - \offset)^+]$. This relies on the relation that $\E[G f(\ve{X}, \ve{Z})] = 0$ as introduced in Section~\ref{sec:drift-prelim}. In our proofs, we choose $f$ to be piecewise quadratic functions to get bounds on the term 
\[
\E\left[\sumall \siz_i \left( Z_i-\xnom_i\right)\cdot (\normwork-\offset)^+\right].
\]
Since $\sumall \siz_i (Z_i - \xnom_i) = \sumall \siz_i Z_i - n+\slk$, we will be able to bound $\E[(\normwork - \offset)^+]$ if we are able to give an accurate estimate of the number of busy servers $\sumall \siz_i Z_i$ when the normalized work $\normwork \geq \offset$. To get a precise estimate, we exploit the state-space concentration result that says for each $i\in[\njt]$, $X_i$ cannot be much smaller than $\xnom_i$, i.e., $(X_i - \xnom_i)^-$ is small with high probability. Formally, this state-space concentration is established by Lemma~\ref{lem:mminf-bounds},
whose proof uses a sample-path coupling argument and is given in Appendix~\ref{sec:mminf-proof}. 

\begin{restatable}{lemma}{mminf}\label{lem:mminf-bounds}
    Consider the multiserver-job system with $n$ servers. For any nonnegative vector $c = (c_1, \dots, c_\njt)\in\mathbb{R}^{\njt}_+$  independent of $n$, let $\cmax=\max_{i\in[\njt]} c_i$, $\sevmax = \max_{i\in[\njt]} \sev_i$ and let
    \begin{equation*}
        \yc=\sumall c_i \siz_i\left(X_i(\infty)-\xnom_i\right),
    \end{equation*}
    where $\xnom_i = \frac{\arr_i}{\sev_i}$.
    Then we have the three bounds below.
    \begin{enumerate}
        \item For any $K \geq 0$, 
        \begin{equation}
            \Prob\left(\yc \leq -K \right) \leq \exp\left(-\frac{K^2}{2\cmax^2 \sevmax \sysvar}\right).
       \end{equation}
       \item For any $\alpha\ge 0$ and $\beta \geq 0$ such that $\alpha\beta \geq \cmax^2\sevmax \sysvar$ and any $j\ge 0$,
       \begin{equation}
           \Prob\left(\yc \leq - \alpha - \beta j \right) \leq e^{-j}.
       \end{equation}
        \item Let $\yc^-=\max\left\{-\yc,0\right\}$ to be the negative part of $\yc$.  Then
        \begin{equation}
            \E\left[\yc^-\right] \leq \sqrt{\cmax^2 \sevmax \sysvar}.
        \end{equation}
    \end{enumerate}
\end{restatable}

Next, we give the proof sketches of Lemma~\ref{result:workload-lower} and Lemma~\ref{result:workload-upper}. 

\subsubsection*{\textbf{Proof sketch of Lemma~\ref{result:workload-lower} (workload lower bound).}}
Recall that $\E[\normwork] = \offset_1 + \E[(\normwork - \offset_1)^+] - \E[(\normwork - \offset_1)^-]$, for some scalar $\offset_1$ to be specified later. To bound the positive part, we invoke the relation $\E[G f(\ve{X}, \ve{Z})]=0$ for a carefully constructed function $f(x)$
and get
\begin{equation}
     \E\left[\sumall \siz_i \left( Z_i-\xnom_i\right)\cdot (\normwork-\offset_1)^+\right] \geq \sysvar - \Obrac{n\Lm} \Prob\left(\normwork \leq \offset_1 + \frac{\Lm}{\sevmin}\right).
\end{equation}
According to Lemma~\ref{lem:mminf-bounds} (a) with $\yc = \normwork$, we can choose some $\offset_1 = -\Obrac{\sqrt{\sysvar} \log n + \Lm}$ such that the probability on the right hand side is bounded by $\frac{1}{n^2}$. Moreover, observe that $\sumall \siz_i \left( Z_i-\xnom_i\right) \leq n - (n-\slk) = \slk$, we get
\begin{equation}
    \E[(\normwork - \offset_1)^+] \geq \frac{\sysvar}{\slk} \cdot (1-o(1)).
\end{equation}

By Lemma~\ref{lem:mminf-bounds} (c) and the fact that $\offset_1 \leq 0$, we can immediately get that the negative part satisfies
\[
    \E[(\normwork - \offset_1)^-] \leq \E[(\normwork)^-] = \Obrac{\sqrt{\sysvar}}.
\]

Combining the bounds on $\E\left[(\normwork-\offset_1)^+\right]$ and $\E[(\normwork - \offset_1)^-]$ gives
\[
    \E[\normwork] = \frac{\sysvar}{\slk} \cdot (1-o(1)) - \Obrac{\sqrt{\sysvar} \log n + \Lm} = \frac{\sysvar}{\slk} \cdot (1-o(1)),
\]
where the last equality follows from Assumption~\ref{assump:traffic} and \ref{assump:lm-bound}, that is, $\slk = \obrac{\sqrt{\sysvar}/\log n}$ and $\Lm \leq \Lmratio \slk$.

\subsubsection*{\textbf{Proof sketch of Lemma~\ref{result:workload-upper} (workload upper bound).}}
Observe that $E[\normwork] \leq \offset_2 + \E[(\normwork - \offset_2)^+]$, for some $\offset_2$ to be specified later. To bound the positive part $\E[(\normwork - \offset_2)^+]$, we apply the relation $\E[G f(\ve{X}, \ve{Z})] = 0$ to a carefully constructed function $f(x)$ 
\begin{equation}\label{eq:proof:bd1:drift-var-terms}
        \E\left[\sumall \siz_i \left( Z_i-\xnom_i\right)\cdot (\normwork-\offset_2)^+\right]
        \leq \sysvar.
\end{equation}

In addition, we claim that there exists some $\gamma > 0$ such that
\begin{equation}\label{eq:proof:workload-upper:intermediate-goal}
    \gamma \cdot \E[(\normwork - \offset_2)^+] \leq \E\left[\sumall \siz_i \left( Z_i-\xnom_i\right)\cdot (\normwork-\offset_2)^+\right] + o(1).
\end{equation}
To prove this, we observe that the RHS term $\E\left[\sumall \siz_i \left( Z_i-\xnom_i\right)\cdot (\normwork-\offset_2)^+\right]$ is non-negative and the LHS term $\E[(\normwork - \offset_2)^+]$ is non-zero only when 
\begin{equation}
    \sumall \frac{\siz_i}{\sev_i}(X_i - \xnom_i) \geq \offset_2.
\end{equation}
By Lemma~\ref{lem:mminf-bounds} (b) with $c_i = \frac{1}{\sevmin} - \frac{1}{\sev_i}$, we also have the following inequality with probability at least $1 - \frac{1}{n^3}$, 
\begin{equation}\label{eq:proof:workload-upper:with-high-prob}
    \sumall \left(\frac{1}{\sevmin} - \frac{1}{\sev_i}\right) \siz_i (X_i - \xnom_i) \geq - K_2 ,
\end{equation}
for some $K_2 = \Obrac{\sqrt{\sysvar}\log n}$. Adding up the two inequalities above and applying $\wst$-work-conserving property, we get
\[
    \sumall \siz_i (Z_i - \xnom_i)  \geq\min\left(\sevmin(\offset_2 - K_2), \slk- \wst\right).
\]

After handling the low probability event that \eqref{eq:proof:workload-upper:with-high-prob} does not hold, we can show \eqref{eq:proof:workload-upper:intermediate-goal} with $\gamma = \min\left(\sevmin(\offset_2 - K_2), \slk- \wst\right)$. Therefore, 
\begin{align*}
    \E[\normwork] &\leq \offset_2 + \E[(\normwork - \offset_2)^+] \leq \offset_2 + \frac{\sysvar}{\min\left(\sevmin(\offset_2 - K_2), \slk- \wst\right)} + o(1).
\end{align*}
The upper bounds \eqref{eq:result:workload-upper-case1} and \eqref{eq:result:workload-upper-case2} in Lemma~\ref{result:workload-upper} follow once we choose a suitable $\offset_2$. When $\slk = \obrac{\sqrt{\sysvar}/\log n}$, choosing $\offset_2 = K_2 + (\slk - \wst)/\sevmin$ yields $\E[\normwork] \leq \frac{\sysvar}{\slk-\wst}\cdot (1+o(1))$. When $\slk = \Wbrac{\sqrt{\sysvar}/\log n}$, choosing $\offset_2 = K_2 + \Thebrac{\sqrt{\sysvar} / \log n}$ yields $\E[\normwork] \leq \Obrac{\sqrt{\sysvar}\log n}$.

\subsection{Proofs of Lemma~\ref{result:workload-lower} and Lemma~\ref{result:workload-upper}.}\label{sec:app:workload-bounds}
In this section we show the full proof of Lemma~\ref{result:workload-lower} and Lemma~\ref{result:workload-upper}.

\proof{Proof of Lemma~\ref{result:workload-lower}.}
    Recall that in Section~\ref{sec:workload-bounds-preliminary-sketch}, we have shown that $\E\left[\sumall \frac{\siz_i}{\sev_i} Q_i\right] = \E[\normwork]$, where $\normwork$ is the \textit{normalized work} given by $\normwork \triangleq \sumall \frac{\siz_i}{\sev_i} \left(X_i - \xnom_i\right)$ and $\xnom_i \triangleq \frac{\arr_i}{\sev_i}$. Therefore, our goal is equivalent to lower bounding $\E[\normwork]$. To do this, we first perform the following decomposition: 
    \[
        \E[\normwork] = \offset_1 + \E[(\normwork-\offset_1)^+] - \E[(\normwork-\offset_1)^-].
    \]
    where $\offset_1$ is a properly chosen small number to be specified later. We bound the positive part $\E[(\normwork - \offset_1)^+]$ and the negative part $\E[(\normwork - \offset_1)^-]$ separately using different techniques. 
    
    We bound $\E[(\normwork-\offset_1)^+]$ by analyzing the Lyapunov drift of a function $f\colon \R_+^\njt \to \R$ defined below.  Let $\ve{x},\ve{z} \in \R_+^\njt$ denote possible realizations of the state descriptors $\ve{X}(t), \ve{Z}(t)$. Then $f$ is defined as
    \begin{equation*}
        f(\ve{x}) = \varphi_M\left(\fnormwork{\ve{x}} - \offset_1\right)
    \end{equation*}
    where $\fnormwork{\ve{x}} \triangleq \sumall \frac{\siz_i}{\sev_i} (x_i - \xnom_i)$ is a possible realization of $\normwork$, $\varphi_M(s)\colon \R \to \R$ is defined as
    \[
        \varphi_M(s) = 
        \begin{cases}
            0 & \text{ if } s \leq 0\\
            s^2 & \text{ if } 0 < s \leq M\\
            2Ms - M^2 &\text{ if } s > M
        \end{cases},
    \]
    and $M$ is a positive number preventing $\varphi_M(s)$ from growing too fast as $s$ gets large. The derivative of $\varphi_M(s)$ is $\varphi_M'(s) = 2\min\left\{s^+, M\right\}$.

    We will utilize the relation $\E[G f(\ve{X}, \ve{Z})] = 0$, which is implied by Lemma~\ref{lem:bar-master-condition} in Appendix~\ref{sec:app:drift-cond} if we have $\E[\abs{f(\ve{X})}] < \infty$. Since $f(\ve{x})$ grows linearly fast as $\ve{x}$ gets large, $\E[f(\ve{x})]<\infty$ follows if we have $\E[X_i] < \infty$ for all $i\in[\njt]$. On the other hand, there is nothing to prove if $\E[X_i] = \infty$ for some $i$. 
    
    To calculate $Gf(\ve{x}, \ve{z})$, we first decompose the drift $G f(\ve{x},\ve{z})$ formula \eqref{eq:proof:drift-form} in the following way:
    \begin{equation}\label{eq:app:proof:drift-mean-var-decomp-in-lemma1}
        \begin{aligned}
            G f(\ve{x}, \ve{z}) &= \sumall \left(\arr_i - \sev_i z_i\right)\frac{\partial f}{\partial x_i}(\ve{x})\\
            &\mspace{23mu}+ \sumall \arr_i \left(f(\ve{x}+\ve{e}_i) - f(\ve{x}) -  \frac{\partial f}{\partial x_i}(\ve{x})\right)+\sumall\sev_i z_i \left(f(\ve{x}-\ve{e}_i) - f(\ve{x}) + \frac{\partial f}{\partial x_i}(\ve{x})\right). 
        \end{aligned}
    \end{equation}
    It is easy to see that the partial derivatives of $f$ appearing in the first term is given by
    \[
    \frac{\partial f}{\partial x_i}(\ve{x}) = \frac{\partial \fnw}{\partial x_i}(\ve{x}) \varphi_M'\left(\fnormwork{\ve{x}}-\offset_1\right)  =\frac{2\siz_i}{\sev_i} \min\left( (\fnormwork{\ve{x}} - \offset_1)^+ , M\right).
    \]
    To bound the remaining two terms, observe that
    \begin{align*}
        f(\ve{x}+\ve{e}_i) - f(\ve{x}) - \frac{\partial f}{\partial x_i}(\ve{x}) &= \int_{\xi \in [\ve{x}, \ve{x}+\ve{e}_i]} \left(\frac{\partial f}{\partial x_i}(\xi) - \frac{\partial f}{\partial x_i}(\ve{x})\right) d\xi\\ 
        &= \frac{\siz_i}{\sev_i} \int_{\xi \in [\ve{x}, \ve{x}+\ve{e}_i]} \left(\varphi_M'\left(\fnormwork{\xi} - \offset_1 \right) - \varphi_M'\left(\fnormwork{\ve{x}} - \offset_1\right)\right) d\xi\\
        &\geq \frac{2\siz_i}{\sev_i} \int_{\xi \in [\ve{x}, \ve{x}+\ve{e}_i]} \left(\fnormwork{\xi} - \fnormwork{\ve{x}}\right) d\xi \cdot\indibrac{ \offset_1 +  \frac{\Lm}{\sevmin}<\fnormwork{\ve{x}} < \offset_1 + M -  \frac{\Lm}{\sevmin}}\\
        &= \frac{\siz_i^2}{\sev_i^2} \cdot\indibrac{ \offset_1 +  \frac{\Lm}{\sevmin}<\fnormwork{\ve{x}} < \offset_1 + M -  \frac{\Lm}{\sevmin}},
    \end{align*}
    where the inequality is due to the fact that  for any $s_1 \geq s_2$,
    \begin{equation}
        \varphi'_M(s_1) - \varphi'_M(s_2) \geq 2(s_1 - s_2) \cdot \indibrac{0 \leq s_2\leq s_1 \leq M},
    \end{equation}
    which can be verified by brute force calculation. Similarly,
    \begin{align*}
        f(\ve{x} - \ve{e}_i) - f(\ve{x}) +  \frac{\partial f}{\partial x_i}(\ve{x})
        &= \int_{\xi \in [\ve{x}-\ve{e}_i, \ve{x}]} \left(-\frac{\partial f}{\partial x_i}(\xi) + \frac{\partial f}{\partial x_i}(\ve{x})\right) d\xi\\ 
        &= \frac{\siz_i}{\sev_i} \int_{\xi \in [\ve{x}-\ve{e}_i, \ve{x}]} \left(-\varphi_M'\left(\fnormwork{\xi} - \offset_1 \right) + \varphi_M'\left(\fnormwork{\ve{x}} - \offset_1\right)\right) d\xi\\
        &\geq \frac{2\siz_i}{\sev_i} \int_{\xi \in [\ve{x}-\ve{e}_i, \ve{x}]} \left(-\fnormwork{\xi} + \fnormwork{\ve{x}}\right) d\xi \cdot\indibrac{ \offset_1 +  \frac{\Lm}{\sevmin}<\fnormwork{\ve{x}} < \offset_1 + M -  \frac{\Lm}{\sevmin}}\\
        &= \frac{\siz_i^2}{\sev_i^2} \cdot\indibrac{ \offset_1 +  \frac{\Lm}{\sevmin}<\fnormwork{\ve{x}} < \offset_1 + M - \frac{\Lm}{\sevmin}}\\
        &\geq \frac{\siz_i^2}{\sev_i^2}\indibrac{ \offset_1 +  \frac{\Lm}{\sevmin}<\fnormwork{\ve{x}} < \offset_1 + M -  \frac{\Lm}{\sevmin}}.
    \end{align*}
    Plugging the above inequalities into the decomposition of the drift \eqref{eq:app:proof:drift-mean-var-decomp-in-lemma1} and taking expectation on both sides, because $\E[G f(\ve{X},\ve{Z})] = 0$, we have
    \begin{equation}
        \begin{aligned}
             0 &\geq \E\left[ \sumall \siz_i \left(\xnom_i - Z_i\right)\cdot 2\min\left( (\normwork - \offset_1)^+ , M\right) + \sumall \left(\frac{\arr_i \siz_i^2}{\sev_i^2} + \frac{Z_i\siz_i^2}{\sev_i}\right) \indibrac{ \offset_1 +  \frac{\Lm}{\sevmin}<\normwork < \offset_1 + M -  \frac{\Lm}{\sevmin}}\right].
        \end{aligned}    
    \end{equation}
    The same inequality still holds when we let $M\to\infty$ inside the expectation, 
    \begin{equation}
        0 \geq \E\left[ \sumall \siz_i \left(\xnom_i - Z_i\right)\cdot 2(\normwork-\offset_1)^+ + \sumall \left(\frac{\arr_i \siz_i^2}{\sev_i^2} + \frac{Z_i\siz_i^2}{\sev_i}\right) \indibrac{\normwork>\offset_1 +  \frac{\Lm}{\sevmin}}\right].
    \end{equation}
    Here we are implicitly exchanging $\lim_{M\to\infty}$ and $\E$. This is legal because the random variable inside the expectation is dominated by another random variable $ 2n\Lm(\normwork-\offset_1)^+ + 2n\Lm / \sevmin$ with a finite expectation, so we can apply the dominated convergence theorem. Using the facts that $\E\left[\sumall \left(\frac{\arr_i \siz_i^2}{\sev_i^2} + \frac{Z_i\siz_i^2}{\sev_i}\right) \right] = 2\sysvar$ and  $\sumall \left(\frac{\arr_i \siz_i^2}{\sev_i^2} + \frac{Z_i\siz_i^2}{\sev_i}\right) \leq \frac{2n\Lm}{\sevmin}$, we get
    \begin{equation}
        \E\left[\sumall \siz_i \left(Z_i - \xnom_i\right)\cdot (\normwork-\offset_1)^+ \right] \geq \sysvar -\frac{n\Lm}{\sevmin} \cdot \Prob\left(\normwork\leq\offset_1 +  \frac{\Lm}{\sevmin}\right).
    \end{equation}
    Consider Lemma~\ref{lem:mminf-bounds} with $\yc= \normwork$, $c_i = \frac{1}{\sev_i}$, $K = 2\sqrt{\sevmax \sysvar \log n} / \sevmin $. Then we have
    $
    \Prob\left(\normwork \leq -K \right) \leq \frac{1}{n^2}.
    $ 
    Note that this choice of $K$ yields $K= \Obrac{\sqrt{\sysvar}\log n}$. We take $\offset_1 = - K - \frac{\Lm}{\sevmin}$, then $\Prob\left(\normwork \leq \offset_1 +  \frac{\Lm}{\sevmin} \right) \leq \frac{1}{n^2}$. Moreover, observe that $\sumall \siz_i \left( Z_i-\xnom_i\right) \leq n - (n-\slk) = \slk$, we get
    \begin{equation}\label{eq:lb-proof:positive-part}
        \E\left[(\normwork-\offset_1)^+ \right] \geq \frac{\sysvar}{\slk}\cdot (1-o(1)).
    \end{equation}
        
    The lower bound for negative part $\E[(\normwork-\offset_1)^-]$ follows from Lemma~\ref{lem:mminf-bounds}, which says that $E[(\normwork)^-] = \Obrac{\sqrt{\sysvar}}$. Because $s\mapsto s^-$ is a monotonically decreasing function and $\offset_1 \leq 0$,
    \begin{equation}\label{eq:lb-proof:negative-part}
        \E[(\normwork-\offset_1)^-] \leq \E[(\normwork)^-] = \Obrac{\sqrt{\sysvar}}.
    \end{equation}
    
    Therefore, combining \eqref{eq:lb-proof:positive-part},  \eqref{eq:lb-proof:negative-part} and the fact that  $\offset_1 = - \Obrac{\sqrt{\sysvar} \log n + \Lm}$, we have
    \[
        \E[\normwork] = \frac{\sysvar}{\slk} \cdot (1-o(1)) - \Obrac{\sqrt{\sysvar} \log n + \Lm} = \frac{\sysvar}{\slk} \cdot (1-o(1)),
    \]
    This finishes the proof of Lemma~\ref{result:workload-lower}.
\endproof

\proof{Proof of Lemma~\ref{result:workload-upper}.}
Recall that in Section~\ref{sec:workload-bounds-preliminary-sketch} we have shown that $\E\left[\sumall \frac{\siz_i}{\sev_i} Q_i\right] = \E[\normwork]$, where $\normwork$ is the \textit{normalized work} given by $\normwork \triangleq \sumall \frac{\siz_i}{\sev_i} \left(X_i - \xnom_i\right)$ and $\xnom_i \triangleq \frac{\arr_i}{\sev_i}$. Moreover, we have
\begin{equation}
    \E[\normwork] \leq \offset_2 + \E[(\normwork-\offset_2)^+],
\end{equation}
for any number $\offset_2$. Next, we will bound the term $\E[(\normwork-\offset_2)^+]$ for a suitably chosen $\offset_2$.

We bound $\E[(\normwork-\offset_2)^+]$ by analyzing the Lyapunov drift of a function $f\colon \R_+^\njt \to \R$ defined below.  Let $\ve{x},\ve{z} \in \R_+^\njt$ denote possible realizations of the state descriptors $\ve{X}(t), \ve{Z}(t)$. Then $f$ is defined as
\begin{equation*}
    f(\ve{x}) = \varphi\left(\fnormwork{\ve{x}} - \offset_2\right)
\end{equation*}
where $\fnormwork{\ve{x}} \triangleq \sumall \frac{\siz_i}{\sev_i} (x_i - \xnom_i)$ is a possible realization of $\normwork$ and $\varphi(s) = \left(s^+\right)^2$.

This proof relies on the relation that $\E[G f(\ve{X}, \ve{Z})] = 0$, which is justified by Lemma~\ref{lem:polynomial-zero-drift} in Appendix~\ref{sec:app:drift-cond}. To calculate $G f(\ve{x},\ve{z})$, we decompose the drift formula \eqref{eq:proof:drift-form} in the following way:
\begin{equation}\label{eq:app:proof:drift-mean-var-decomp}
    \begin{aligned}
        G f(\ve{x}, \ve{z}) &= \sumall \left(\arr_i - \sev_i z_i\right)\frac{\partial f}{\partial x_i}(\ve{x})\\
        &\mspace{23mu}+ \sumall \arr_i \left(f(\ve{x}+\ve{e}_i) - f(\ve{x}) -  \frac{\partial f}{\partial x_i}(\ve{x})\right)+\sumall\sev_i z_i \left(f(\ve{x}-\ve{e}_i) - f(\ve{x}) + \frac{\partial f}{\partial x_i}(\ve{x})\right). 
    \end{aligned}
\end{equation}
It is easy to see that
\begin{equation*}
\frac{\partial f}{\partial x_i}(\ve{x}) = \frac{\partial r}{\partial x_i}(\ve{x}) \varphi'\left(\fnormwork{\ve{x}} - \offset_2\right)  =\frac{2\siz_i}{\sev_i} (\fnormwork{\ve{x}} - \offset_2)^+.
\end{equation*}
The remaining two terms are bounded by constants independent of $\ve{x}$ and $\ve{z}$ in the following way. 
\begin{align*}
    f(\ve{x}+\ve{e}_i) - f(\ve{x}) - \frac{\partial f}{\partial x_i}(\ve{x}) &= \int_{\xi \in [\ve{x}, \ve{x}+\ve{e}_i]} \left(\frac{\partial f}{\partial x_i}(\xi) - \frac{\partial f}{\partial x_i}(\ve{x})\right) d\xi\\ 
    &= \frac{\siz_i}{\sev_i} \int_{\xi \in [\ve{x}, \ve{x}+\ve{e}_i]} \left(\varphi'\left(\fnormwork{\xi}-\offset_2\right) - \varphi'\left(\fnormwork{\ve{x}}-\offset_2\right)\right) d\xi\\
    &\leq \frac{2\siz_i}{\sev_i} \int_{\xi \in [\ve{x}, \ve{x}+\ve{e}_i]} \left(\fnormwork{\xi} - \fnormwork{\ve{x}}\right) d\xi \\ 
    &= \frac{\siz_i^2}{\sev_i^2}, 
 \end{align*}
where the inequality is because $\varphi'(s) = 2s^+$ is $2$-Lipschitz continuous. 
Similarly,
\begin{equation*}
    f(\ve{x} - \ve{e}_i) - f(\ve{x}) +  \frac{\partial f}{\partial x_i}(\ve{x}) \leq \frac{\siz_i^2}{\sev_i^2}. 
\end{equation*}
Plugging in the inequalities above back to the decomposition of the drift \eqref{eq:app:proof:drift-mean-var-decomp} and taking expectation on both sides, because $\E[G f(\ve{X},\ve{Z})] = 0$, we have
\begin{equation}
    \begin{aligned}
         0 &\leq \E\left[ \sumall \siz_i \left(\xnom_i - Z_i\right)\cdot 2(\normwork-\offset_2)^+\right]+ \E\left[\sumall \left(\frac{\arr_i \siz_i^2}{\sev_i^2} + \frac{Z_i\siz_i^2}{\sev_i}\right)\right]. 
    \end{aligned}    
\end{equation}
Observe that because $\E[Z_i] = \frac{\arr_i}{\sev_i}$, the second term can be straightforwardly computed as 
\begin{align}
       \E \left[\sumall \left(\frac{\arr_i \siz_i^2}{\sev_i^2} + \frac{Z_i \siz_i^2}{\sev_i}\right)\right] = 2\sumall \frac{\arr_i\siz_i^2}{\sev_i^2} = 2\sysvar. 
\end{align}
Rearranging the terms, we get the following key equation: for any number $\offset_2$, 
\begin{equation}\label{eq:app:proof:bd1:mean-var-terms}
    \E\left[ \sumall \siz_i \left(Z_i - \xnom_i\right)\cdot (\normwork-\offset_2)^+\right] \leq \sysvar.
\end{equation}
Now suppose we are able to show that
\begin{equation}\label{eq:app:proof:drift-goal}
    \gamma \cdot \E\left[ (\normwork-\offset_2)^+\right] \leq \E\left[ \sumall \siz_i \left(Z_i - \xnom_i\right) \cdot (\normwork-\offset_2)^+ \right] + o(1),
\end{equation}
for some $\gamma > 0$, then by \eqref{eq:app:proof:bd1:mean-var-terms}, $\E\left[ (\normwork-\offset_2)^+\right] \leq \frac{\sysvar + o(1)}{\gamma}$, so $\E[\normwork] \leq \offset_2 + \frac{\sysvar +o(1)}{\gamma}$.

We devote the remainder of this proof to proving \eqref{eq:app:proof:drift-goal}.
The idea here is to use the following two events to further partition the probability space
\begin{align*}
\mathcal{E}_1 = \left\{\sumall \left(\tfrac{1}{\sevmin} - \tfrac{1}{\sev_i} \right)\siz_i(X_i - \xnom_i) > -K_2\right\}, \quad
\mathcal{E}_2 = \left\{\normwork \leq \frac{n}{\sevmin}\right\},
\end{align*}
where $K_2$ is a suitable number such that $\mathcal{E}_1$ happens with high probability.
We break the term on the left of \eqref{eq:app:proof:bd1:mean-var-terms} based on the three cases $\mathcal{E}_1$, $\mathcal{E}_1^c \cap \mathcal{E}_2$ and $\mathcal{E}_1^c \cap \mathcal{E}_2^c$ and analyze them separately.
\begin{align}
    \E\left[ \sumall \siz_i\left(Z_i - \xnom_i\right)\cdot (\normwork-\offset_2)^+\right] &=  \E\left[ \sumall \siz_i\left(Z_i - \xnom_i\right)\cdot (\normwork-\offset_2)^+ \indi_{\{\normwork>\offset_2, \mathcal{E}_1\}} \right]\label{eq:app:proof:bd1:indicator-decomp1}\\
    &\mspace{5mu}+ \E\left[\sumall \siz_i\left(Z_i - \xnom_i\right)\cdot (\normwork-\offset_2)^+\indi_{\{\normwork>\offset_2,\mathcal{E}_1^c,\mathcal{E}_2\}}\right]\label{eq:app:proof:bd1:indicator-decomp2} \\
    &\mspace{5mu}+ \E\left[\sumall \siz_i\left(Z_i - \xnom_i\right)\cdot (\normwork-\offset_2)^+\indi_{\left\{\normwork>\offset_2,\mathcal{E}_1^c, \mathcal{E}_2^c\right\}}\right]\label{eq:app:proof:bd1:indicator-decomp3}.
\end{align}

\textbf{Case 1: $\mathcal{E}_1$ happens.} Observe that the term in \eqref{eq:app:proof:bd1:indicator-decomp1} is non-zero only when $\normwork >\offset_2$ and event $\mathcal{E}_1$ happens, which implies that
\begin{equation*}
    \sumall \frac{\siz_i}{\mu_i} (X_i - \xnom_i) > \offset_2,
\end{equation*}
\begin{equation*}
    \sumall \left(\frac{1}{\sevmin} - \frac{1}{\sev_i} \right)\siz_i(X_i - \xnom_i) > - K_2.
\end{equation*}
Adding up the above two inequalities and rearranging the terms yield
\begin{equation}
    \sumall \siz_i(X_i  - \xnom_i) > \sevmin (\offset_2 - K_2).
\end{equation}
When the above inequality holds, we can invoke the definition of $\wst$-work-conserving policy and the fact that $\sumall \siz_i \xnom_i = n-\slk$ to get
\begin{equation}
    \begin{aligned}
         \sumall \siz_i\left(Z_i - \xnom_i\right)      &\geq \min\left(\sumall \siz_i X_i, n-\wst\right) - \sumall\siz_i\xnom_i \\
         &= \min\left(\sumall \siz_i(X_i - \xnom_i), \slk-\wst\right) \geq \min\left(\sevmin(\offset_2 - K_2), \slk - \wst\right).
    \end{aligned}
\end{equation}
\begin{equation}
    \begin{aligned}
       &\E\left[ \sumall \siz_i\left(Z_i - \xnom_i \right)\cdot (\normwork-\offset_2)^+ \indi_{\left\{\normwork > \offset_2,\mathcal{E}_1\right\}}\right] \\
       &\geq \min\left(\sevmin(\offset_2 - K_2), \slk - \wst\right)\cdot\E\left[\left(\normwork - \offset_2\right)^+\indi_{\left\{\normwork > \offset_2,\mathcal{E}_1\right\}}\right] .
    \end{aligned}\label{eq:app:proof:bd1:drift-part1}
\end{equation}

\textbf{Case 2: $\mathcal{E}_1^c \cap \mathcal{E}_2$ happens.} To bound the term in \eqref{eq:app:proof:bd1:indicator-decomp2}, we need to analyze the probability of event $\mathcal{E}_1^c$. Consider Lemma~\ref{lem:mminf-bounds} with $\yc= \sumall \left(\frac{1}{\sevmin} - \frac{1}{\sev_i} \right)\siz_i(X_i - \xnom_i)$, $c_i = \frac{1}{\sevmin} - \frac{1}{\sev_i}$, $\alpha = \beta= \sqrt{\frac{\sevmax}{\sevmin^2}\sysvar}$, $j=3\log n$. It can be verified that
$\alpha\beta \geq \cmax^2\sevmax\sysvar$.  Let $K_2=\alpha+\beta j$. Then we have
\begin{align*}
\Prob\left(\mathcal{E}_1^c\right)=\Prob\left(\sumall \left(\frac{1}{\sevmin} - \frac{1}{\sev_i} \right)\siz_i(X_i - \xnom_i) \leq -K_2 \right)\le \frac{1}{n^3}.
\end{align*}
Note that this choice of $K_2$ satisfies $K_2= O(\sqrt{\sysvar}\log n)$. 

With the upper bound $\Prob\left(\mathcal{E}_1^c\right)\le\frac{1}{n^3}$, we bound the term in \eqref{eq:app:proof:bd1:indicator-decomp2} in the following way. Observe that $\sumall\siz_i(Z_i - \xnom_i) \geq  -n$.  Further, when the event $\mathcal{E}_2$ occurs, $\normwork\le\frac{n}{\sevmin}$, thus $0 \leq (\normwork-\offset_2)^+\le\frac{n}{\sevmin}$.  Therefore,
\begin{equation*}
    \mspace{23mu}\E\left[\sumall\siz_i(Z_i - \xnom_i)\cdot(\normwork - \offset_2)^+\indi_{\left\{\normwork > \offset_2, \mathcal{E}_1^c, \mathcal{E}_2\right\}}\right] 
    \geq -n\cdot \frac{n}{\sevmin}\cdot \Prob\left(\mathcal{E}_1^c\right)
    \geq -\frac{1}{\sevmin n},
\end{equation*}
\begin{equation*}
    \mspace{23mu}\E\left[(\normwork - \offset_2)^+ \indi_{\left\{\normwork > \offset_2, \mathcal{E}_1^c, \mathcal{E}_2\right\}}\right]
    \leq \frac{n}{\sevmin}\cdot \Prob\left(\mathcal{E}_1^c\right)
    \leq \frac{1}{\sevmin n^2}.
\end{equation*}
We can rearrange the above inequalities into a similar form as \eqref{eq:app:proof:bd1:drift-part1}:
\begin{equation}\label{eq:app:proof:bd1:drift-part2}
    \begin{aligned}
        &\mspace{23mu}\E\left[\sumall\siz_i(Z_i - \xnom_i)\cdot(\normwork - \offset_2)^+\indi_{\left\{\normwork > \offset_2, \mathcal{E}_1^c, \mathcal{E}_2\right\}}\right]\\
        &\geq \min\left(\sevmin(\offset_2 - K_2), \slk - \wst\right) \cdot \E\left[(\normwork - \offset_2)^+ \indi_{\left\{\normwork > \offset_2, \mathcal{E}_1^c, \mathcal{E}_2\right\}}\right] - \frac{n+\slk}{\sevmin n^2},
    \end{aligned}
\end{equation}
where we have used the fact that $\min\left(\sevmin(\offset_2 - K_2), \slk - \wst\right) \leq \slk$.

\textbf{Case 3: $\mathcal{E}_1^c \cap \mathcal{E}_2^c$ happens.} Lastly, we bound the term in (\ref{eq:app:proof:bd1:indicator-decomp3}). Observe that $\mathcal{E}_2^c$ implies that $\sumall \siz_i X_i \geq n$, so an $\wst$-work-conserving policy will make sure that $\sumall \siz_i Z_i \geq n-\wst$. Therefore, 
\begin{equation}\label{eq:app:proof:bd1:drift-part3}
    \begin{aligned}
      \mspace{23mu}\E\left[\sumall\siz_i(Z_i - \xnom_i)\cdot(\normwork - \offset_2)^+\indi_{\left\{\normwork > \offset_2, \mathcal{E}_1^c, \mathcal{E}_2^c\right\}}\right]
      \geq (\slk - \wst)\cdot\E\left[ (\normwork - \offset_2)^+\indi_{\left\{\normwork > \offset_2, \mathcal{E}_1^c, \mathcal{E}_2^c\right\}}\right].
     \end{aligned}
\end{equation}

Combining the results in three cases \eqref{eq:app:proof:bd1:drift-part1} \eqref{eq:app:proof:bd1:drift-part2}
\eqref{eq:app:proof:bd1:drift-part3}, we have
\begin{equation}
\E\left[ \sumall \siz_i \left(Z_i - \xnom_i\right) \cdot (\normwork-\offset_2)^+ \right] \geq \min\left(\sevmin(\offset_2 - K_2), \slk - \wst\right) \cdot \E\left[ (\normwork-\offset_2)^+\right] - \frac{n+\slk}{\sevmin n^2}.
\end{equation}
Therefore, we have shown \eqref{eq:app:proof:drift-goal} with $\gamma = \min\left(\sevmin(\offset_2 - K_2), \slk - \wst\right)$. By \eqref{eq:app:proof:bd1:mean-var-terms}, we conclude that
\begin{equation}
    \begin{aligned}
        \E[\normwork] &\leq \offset_2 + \E\left[(\normwork - \offset_2)^+\right] \leq \offset_2 +  \frac{\sysvar}{\min\left(\sevmin(\offset_2 - K_2), \slk - \wst\right)} + O\left(\frac{1}{n}\right).
    \end{aligned}
\end{equation}

When $\slk= \obrac{\frac{\sqrt{\sysvar}}{\log n}}$, we take $\offset_2 = K_2 + \frac{\slk - \slk'}{\sevmin}$ and get
\begin{align}
    \E[\normwork] &\leq K_2 + \frac{\slk - \slk'}{\sevmin} + \frac{\sysvar}{\slk - \wst} + \Obrac{\frac{1}{n}} = \frac{\sysvar}{\slk - \wst} \cdot \left(1 + o(1)\right),\nonumber
\end{align}
where we have used the fact that $K_2 = O(\sqrt{\sysvar}\log n)$ and $\slk = o\left(\frac{\sqrt{\sysvar}}{\log n}\right)$.

When $\slk= \Wbrac{\frac{\sqrt{\sysvar}}{\log n}}$, we have $\slk \geq C \frac{\sqrt{\sysvar}}{\log n}$ for some $C > 0$ and $n$ large enough. Because $\slk - \wst \geq \Lmratio \slk$, we get $\slk - \wst \geq \Lmratio  C \frac{\sqrt{\sysvar}}{\log n}$ for $n$ large enough. Taking $\offset_2 = K_2 + \frac{\Lmratio C }{\sevmin}\frac{\sqrt{\sysvar}}{\log n}$ yields $\min\left(\sevmin(\offset_2 - K_2), \slk - \wst\right) \geq  \Lmratio C  \frac{\sqrt{\sysvar}}{\log n}$, so
\begin{align*}
    \E[\normwork] &\leq K_2 + \frac{\Lmratio C  \sqrt{\sysvar}}{\sevmin \log n} + \frac{\sqrt{\sysvar}\log n}{\Lmratio C} + \Obrac{\frac{1}{n}} = \Obrac{\sqrt{\sysvar}\log n}.
\end{align*}
This finishes the proof.
\endproof

\section{Proof of Theorem~\ref{result:waiting-time-fcfs} (waiting times under FCFS)}\label{sec:proof:fcfs}
\subsection{Proof overview and preliminaries}
We prove Theorem~\ref{result:waiting-time-fcfs} in this section. We first present some intuition and preliminaries. Then we give a proof the theorem. Finally we give the proofs of the two lemmas used in the proof, Lemma~\ref{lem:fcfs:independence} and Lemma~\ref{lem:fcfs:sandwich}.

The analysis of FCFS is based on the intuition that, under FCFS, the number of type~$i$ jobs in the queue is approximately proportional to its arrival rate $\arr_i$, i.e.    
\begin{equation}\label{eq:fcfs-heuristics}
    \E[Q_i] \approx \frac{\arr_i}{\arr} \E[Q_\Sigma],
\end{equation}
where $Q_\Sigma \triangleq \sumall Q_i$ is the total queue length. Using this relation, we can easily convert from the expected workload to mean waiting time:
\begin{equation}
    \sumall \frac{\siz_i}{\sev_i} \E[Q_i] \approx \sumall \frac{\arr_i \siz_i}{\sev_i} \frac{1}{\arr} \E[Q_\Sigma] = (n-\slk) \E\left[\wait\right],
\end{equation}
where the second equality is due to Little's law and the fact that $ \sumall \frac{\arr_i \siz_i}{\sev_i} = n-\slk$.

This intuition is formalized by considering a \textit{Modified-FCFS} policy, which serves the jobs in a FCFS order, but will not serve the next job until there are at least $\Lm$ idle servers. In Lemma~\ref{lem:fcfs:independence}, we prove that under Modified-FCFS, \eqref{eq:fcfs-heuristics} holds exactly, so the expected queue length of each type of jobs can be computed using the above argument. Moreover, when $n$ is large, we expect Modified-FCFS to be a good approximation of FCFS. In fact, we can define an \textit{upper bounding system} as the multiserver-job system with $n$ servers under Modified-FCFS, and define a \textit{lower bounding system} as the multiserver-job system with $n+\Lm$ servers under Modified-FCFS. In Lemma~\ref{lem:fcfs:sandwich}, we prove that the queue length of each type of jobs in the original system is sandwiched between those in the two modified systems. Since the two modified systems themselves are very close to each other, we get a tight characterization of FCFS.   

\begin{lemma}\label{lem:fcfs:independence}
    For both the upper bounding system and the lower bounding system, we have
    \begin{align}
        \E[\sysup{Q}_i(\infty)] &= \frac{\arr_i}{\arr}\E[ \sysup{Q}_{\Sigma}(\infty) ],\quad  \forall i\in [\njt], \label{eq:proof:system2:queue-fraction}\\
        \E[\syslow{Q}_i(\infty)] &= \frac{\arr_i}{\arr}\E[ \syslow{Q}_{\Sigma}(\infty) ],\quad  \forall i\in [\njt].\label{eq:proof:system3:queue-fraction}
    \end{align}
    where $\sysup{Q}_i(t)$ and $\syslow{Q}_i(t)$ denote the queue lengths of type $i$ jobs in the upper and lower bounding systems, respectively; $\sysup{Q}_{\Sigma}(t) \triangleq \sum_{i=1}^\njt \sysup{Q}_i(t)$ and $\syslow{Q}_{\Sigma}(t) \triangleq \sum_{i=1}^\njt \syslow{Q}_i(t)$ are the total queue lengths.
\end{lemma}

\begin{lemma}\label{lem:fcfs:sandwich}
    \begin{equation}\label{eq:proof:fcfs:sandwich}
        \E[\syslow{Q}_i(\infty)] \leq \E\left[Q_i(\infty)\right] \leq \E[\sysup{Q}_i(\infty)] \quad  \forall i\in [\njt].
    \end{equation}
    where $\sysup{Q}_i(t)$ and $\syslow{Q}_i(t)$ denote the queue lengths of type $i$ jobs in the upper and lower bounding systems, respectively.
\end{lemma}

\subsection{Proof of Theorem~\ref{result:waiting-time-fcfs}}
\begin{proof}[Proof of Theorem~\ref{result:waiting-time-fcfs}.]
Recall that by Little's law, $\E\big[\wait_i\big] = \frac{1}{\arr_i} \E[Q_i]$ and $\E\big[\wait\big] = \frac{1}{\arr} \sumall \E[Q_i]$.
To characterize $\E[Q_i]$, Lemma~\ref{lem:fcfs:sandwich} suggests that we only need to characterize $\E[\syslow{Q}_i]$ and $\E[\sysup{Q}_i]$, the queue lengths in the lower bounding system and the upper bounding system, for each $i\in[\njt]$.

Observe that the lower bounding system has $n+\Lm$ servers, with slack capacity $\slk+\Lm$. Because $\slk + \Lm \leq (1+\Lmratio)\slk = \obrac{\frac{\sqrt{\sysvar}}{\log n}}$, we can apply Lemma~\ref{result:workload-lower} and Lemma~\ref{lem:fcfs:independence} to get 
\[
    \frac{1}{\arr} \E[\syslow{Q}_{\Sigma}]\cdot\sumall \frac{\arr_i \siz_i}{\sev_i} = \sumall \frac{\siz_i}{\sev_i} \E[\syslow{Q}_i] \geq \frac{\sysvar}{\slk+\Lm} \cdot (1-o(1)). 
\]
Recall the facts introduced in Section~\ref{sec:model} that $\slk = \obrac{n}$ and $\slk =  n - \sumall \frac{\arr_i \siz_i}{\sev_i}$, so we have $\sumall \frac{\arr_i \siz_i}{\sev_i} = n \cdot(1-o(1))$. Therefore, the above equation implies that $\E[\syslow{Q}_{\Sigma}] \geq  \frac{\arr \sysvar}{n(\slk+\Lm)} \cdot (1-o(1))$. By Lemma~\ref{lem:fcfs:independence}, 
\begin{equation}
    \E[\syslow{Q}_i] \geq  \frac{\arr_i\sysvar}{n(\slk+\Lm)} \cdot (1-o(1)) \quad \forall i\in[\njt].
\end{equation}

For the upper bounding system, observe that it has $n$ server and slack capacity $\slk$, and operates under a $\Lm$-work-conserving policy. Applying the workload upper bound in Lemma~\ref{result:workload-upper} yields
\begin{equation*}
    \frac{1}{\arr} \sumall \frac{ \siz_i}{\sev_i} \E[\sysup{Q}_i] \leq \frac{\sysvar}{\slk-\Lm} \cdot (1+o(1)).
\end{equation*}
Following a similar argument as in the lower bounding system, we have \begin{equation}
    \E[\sysup{Q}_i] \leq \frac{\arr_i \sysvar}{n(\slk-\Lm)}\cdot (1+o(1)) \quad \forall i\in[\njt].
\end{equation}

Therefore, by Lemma~\ref{lem:fcfs:sandwich}, we have
\begin{equation}
  \frac{\arr_i\sysvar}{n(\slk+\Lm)} \cdot (1-o(1)) \leq \E[Q_i] \leq \frac{\arr_i \sysvar}{n(\slk-\Lm)}\cdot (1+o(1)) \quad \forall i\in[\njt],
\end{equation}
which implies the waiting time bounds in Theorem~\ref{result:waiting-time-fcfs}.
\end{proof}

\subsection{Proof of Lemma~\ref{lem:fcfs:independence} and Lemma~\ref{lem:fcfs:sandwich}}
The lemmas are proved using some sample-path coupling arguments. Before showing the proofs, we first give a construction of the sample paths of a multiserver-job system. We index all the jobs by their order of arriving to the system. Let $A_0(k)$, $S_0(k)$, $C(k)$ be three sequences of independent random variables, where $A_0(k)$ follows $\exp(\arr)$ distribution, $S_0(k)$ follows $\exp(1)$ distribution, and $C(k)=i$ with probability $\arr_i/\arr$ for $i\in[\njt]$. For the $k$-th job, we let it arrive to the system at the time $A(k) = \sum_{j=1}^k A_0(k)$, with service time $S(k) = S_0(k) / \sev_{C(k)}$ and server need $L(k) = \siz_{C(k)}$. 

We can simulate the system dynamic under FCFS or Modified-FCFS given the realization of the jobs. Let the $k$-th job's waiting time be $W(k)$, then the $k$-th job starts its service at time $A(k)+W(k)$, occupying $L(k)$ servers and leaves the system at $A(k)+W(k)+S(k)$.

We couple the sample paths of the original system, the upper bounding system, and the lower bounding system by making them share the same realization of jobs, i.e., the $k$-the job of the three system arrive at the same time, with the same service times and server needs. 

Now we are ready to present the proof. Note that quantities with superscript $(U)$ corresponds to the upper bounding system, while quantities with $(L)$ belongs to the lower bounding system.



\begin{proof}[Proof of Lemma~\ref{lem:fcfs:independence}.]
In the upper bounding system, the queue length of type $i$ jobs at time $t$ is given by $
        \sysup{Q}_i(t) = \sum_{k=1}^\infty \indi\{C(k) = i, A(k) \leq t, A(k) + \sysup{W}(k) > t\}.
$ 
Taking expectation, we get 
\begin{equation} 
    \begin{aligned}
        \E[\sysup{Q}_i(t)] = \sum_{k=1}^\infty \frac{\arr_i}{\arr} \Prob\Big(A(k) \leq t, A(k) + \sysup{W}(k) > t  \mid C(k) = i\Big) . 
    \end{aligned}\label{eq:proof:modified-fcfs:queue-fraction-intermediate}
\end{equation}

Let the $\sigma$-algebra $\mathcal{F}_k = \sigma((A_0(j), S_0(j), C(j))\mid j\leq k)$. From the definitions, it is easy to see that  
\[
    \left(S(j), A(j), L(j), \sysup{W}(j) \right)\in \mathcal{F}_k \quad \text{ for }  j\leq k.
\]
Observe that the following two events are identical, 
\begin{equation} \label{eq:proof:fcfs:independence-long-ineq}
        \{A(k)\leq t, A(k)+\sysup{W}(k) > t\} = \{A(k) \leq t, \sum_{j\in \sysup{\mathcal{\ve{Z}}}(t;k-1)} L(j) > n - \Lm\},
\end{equation}
where $\sysup{\mathcal{\ve{Z}}}(t;k) \triangleq \{j\mid j\leq k, A(j) + \sysup{W}(j) + S(j) > t\}$ is the set of jobs with index $j\leq k$ who has not finished service by time $t$ in the upper bounding system. 
Moreover, the RHS event is obviously in the $\sigma$-algebra generated by $\mathcal{F}_{k-1}$ and $A_0(k)$, which makes it independent of $C(k)$. This fact helps us get rid of the conditioning on the RHS of (\ref{eq:proof:modified-fcfs:queue-fraction-intermediate}):
\begin{equation}
    \E[\sysup{Q}_i(t)]= \frac{\arr_i}{\arr}\sum_{k=1}^\infty  \Prob\left(A(k) \leq t, A(k) + \sysup{W}(k) > t\right),
\end{equation}
\begin{equation}
    \E[\sysup{Q}_{\Sigma}(t)] = \sum_{i=1}^\njt \E \sysup{Q}_i(t) =\sum_{k=1}^\infty  \Prob\left(A(k) \leq t, A(k) + \sysup{W}(k) > t\right).
\end{equation}
Dividing the above two inequalities, we get (\ref{eq:proof:system2:queue-fraction}).

The proof of (\ref{eq:proof:system3:queue-fraction}) is almost identical with all superscript $(U)$ replaced by superscript $(L)$, and equation \eqref{eq:proof:fcfs:independence-long-ineq} replaced by 
\begin{equation} 
    \begin{aligned}
        \{A(k)\leq t, A(k)+\syslow{W}(j) > t\} = \{A(k) \leq t, \sum_{j\in \syslow{\mathcal{\ve{Z}}}(t;k-1)} \syslow{L}(j) > n\},
    \end{aligned}
\end{equation}
where $\syslow{\mathcal{\ve{Z}}}(t;k) \triangleq \{j\mid j\leq k, A(j) + \syslow{W}(j) + S(j) > t\}$.
\end{proof}


\begin{proof}[Proof of Lemma~\ref{lem:fcfs:sandwich}.]
We prove the following claim by doing induction on $k$:
\begin{equation}
    W(k) \leq \sysup{W}(k) \quad \forall k=1,2,\dots. \label{eq:proof:fcfs:induction-on-wait}
\end{equation}
Case $k=1$ is trivial because $W(1) = \sysup{W}(1) = 0$. Suppose that we have proved the cases of $1,2,\dots, k$. For case $k+1$, to get a contradiction, suppose that $W(k+1) > \sysup{W}(k+1)$. Let $T = A(k+1) + \sysup{W}(k+1)$, then the $(k+1)$-th job starts service at time $T$ in the upper bounding system, but does not start service until $A(k+1) + W(k+1)$ in the original system. This implies that at the moment before time $T$, the original system has more than $n - L(k+1)$ busy servers, while the upper bounding system has at most $n - \Lm$ busy servers. Because $n - L(k+1) \geq n - \Lm$,
\begin{equation}
    \sum_{j\in \mathcal{\ve{Z}}(T;k)} L(j) > \sum_{j\in \sysup{\mathcal{\ve{Z}}}(T;k)} L(j), \label{eq:proof:fcfs:busy-server-inequality}
\end{equation}
where
$
\mathcal{\ve{Z}}(T;k) = \{j\mid j\leq k, A(j) + W(j) + S(j) > T\},~ 
\sysup{\mathcal{\ve{Z}}}(T;k) = \{j\mid j\leq k, A(j) + \sysup{W}(j) + S(j) > T\}.
$ However, by induction hypothesis, $\mathcal{\ve{Z}}(T;k) \subseteq \sysup{\mathcal{\ve{Z}}}(T;k)$, which contradicts (\ref{eq:proof:fcfs:busy-server-inequality}). Therefore $W(k+1) \leq \sysup{W}(k+1)$. By induction, (\ref{eq:proof:fcfs:induction-on-wait}) is true.

Observe that the number of type~$i$ jobs in the queues at time $t$ is 
\begin{align}
    Q_i(t) &= \sum_{k=1}^\infty \indi\{C(k) = i, A(k)\leq t, A(k) + W(k) > t\} \label{eq:proof:fcfs:q-w}\\
    \sysup{Q}_i(t) &= \sum_{k=1}^\infty \indi\{C(k) = i, A(k)\leq t, A(k) + \sysup{W}(k) > t\}.\label{eq:proof:fcfs:q-w-upper}
\end{align}
by (\ref{eq:proof:fcfs:induction-on-wait}), $Q_i(t) \leq \sysup{Q}_i(t)$ for any $t, i$. Therefore, in steady-state, we have $\E[Q_i] \leq \E[\sysup{Q}_i]$ for any $i$.

For the other inequality in the lemma, we perform a similar argument. We prove the following claim by doing induction on $k$:
\begin{equation}
    W(k) \geq \syslow{W}(k) \quad \forall k=1,2,\dots. \label{eq:proof:fcfs:induction-on-wait2}
\end{equation}
Case $k=1$ is trivial because $W(1) = \syslow{W}(1) = 0$. Suppose we have proved case $1,2,\dots, k$. For case $k+1$, to get a contradiction, suppose $W(k+1) < \syslow{W}(k+1)$. Let $T = A(k+1) + W(k+1)$, then $T< A(k+1) + \syslow{W}(k+1)$. This implies that at time $T$, the $(k+1)$-th job in the original system starts its service, while the $(k+1)$-th job in the lower bounding system is still waiting at the head of line. According to the policy they use, at the moment right before $T$, the original system has at most $n-L(k+1)$ busy servers, while the lower bounding system has $>n$ busy servers. Therefore, 
\begin{equation}
    \sum_{j\in \mathcal{\ve{Z}}(T;k)} L(j) < \sum_{j\in \syslow{\mathcal{\ve{Z}}}(T;k)} L(k), \label{eq:proof:fcfs:busy-server-inequality2}
\end{equation}
where $\syslow{\mathcal{\ve{Z}}}(T;k) = \{j \mid j\leq k, A(j) + \syslow{W}(j) + S(j) > T\}$. 
However, by induction hypothesis, $\syslow{\mathcal{\ve{Z}}}(T;k) \subseteq \mathcal{\ve{Z}}(T;k)$, which contradicts (\ref{eq:proof:fcfs:busy-server-inequality2}). Therefore $W(k+1) \geq \syslow{W}(k+1)$ and \eqref{eq:proof:fcfs:induction-on-wait2} is proved.

Because
\begin{equation}
    \syslow{Q}_i(t) = \sum_{k=1}^\infty \indi\{C(k) = i, A(k)\leq t, A(k) + \syslow{W}(k) > t\}, \label{eq:proof:fcfs:q-w-lower}
\end{equation}
combined with \eqref{eq:proof:fcfs:q-w} \eqref{eq:proof:fcfs:induction-on-wait2},  it follows that $Q_i(t) \geq \syslow{Q}_i(t)$ for all $i,t$. Therefore, in steady-state, we have $\E[\syslow{Q}_i] \leq \E[Q_i]$ for any $i$. This finishes the proof of Lemma~\ref{lem:fcfs:sandwich}.
\end{proof}

\section{Proof of Theorem~\ref{result:waiting-time-lower-bound} (mean waiting time lower bound)}\label{sec:proof:lower}

\begin{proof}[Proof of Theorem~\ref{result:waiting-time-lower-bound}.]
    Recall that by Little's law, $\E\big[\wait\big] = \frac{1}{\arr}\sumall \E[Q_i]$.
    Therefore, it suffices to show a lower bound on the total queue length. We fix a policy in the original system.
    For any $i$ with $\indexmain \leq i \leq \njt$, we consider the $i$-th subsystem by ignoring all job types with index greater than $i$. In the $i$-th subsystem, it is always possible to achieve the same $\E[Q_j]$'s for $j\leq i$ by imitating the service decisions taken by the original system. Therefore, we have  $\sum_{j=1}^i \frac{\siz_j }{\sev_j}\E[Q_j]\ge \frac{\sysvar_i}{\slk_i} \cdot (1-o(1))$, where the right hand side expression is the workload lower bound of the $i$-th subsystem according to Lemma~\ref{result:workload-lower}.
    Then the expected waiting time $\E[\wait]$ is lower-bounded by the optimal value of the following linear programming problem:
    \begin{alignat}{2}
            &\min_{\{q_j\colon j\in[\njt]\}}& \mspace{23mu} &\frac{1}{\arr}\sum_{j=1}^\njt  q_j\nonumber\\
            &\text{subject to}& \quad &\sum_{j=1}^i  \frac{\siz_j}{\sev_j} q_j\geq \frac{\sysvar_i}{\slk_i} \cdot (1-o(1)) \quad \indexmain \leq i \leq \njt \nonumber\\
            &&& q_j \geq 0 \quad \forall j \in [\njt],\nonumber
    \end{alignat}
    where $q_j$ corresponds to  $\E[Q_j]$. It is easy to see that the objective value satisfies
    \begin{align}\label{eq:lb-key-step}
        \frac{1}{\arr}\sum_{j=1}^{\njt}  q_j \geq  \frac{1}{\arr}\frac{\sevmin}{\siz_i}\sum_{j=1}^i \frac{\siz_j}{\sev_j}q_j \geq \frac{\sevmin\sysvar_i}{\arr\siz_i\slk_i} \cdot (1-o(1)),
    \end{align}
    for any $\indexmain \leq i \leq \njt$, where in the first inequality we have used the fact that $\sev_j\geq\sevmin$ and $\siz_j\leq\siz_i$ for any $j\leq i$. Note that when $q_j$ is zero for each $j$ with $j\neq i$, the first inequality becomes an equality up to a constant order factor in terms of $\sev_i$ and $\sevmin$. 
    Because the choice of $i$ with $\indexmain \leq i \leq \njt$ is arbitrary, we have that the optimal value is no less than
    \begin{equation}
       \max_{\indexmain \leq i \leq \njt} \frac{\sevmin \sysvar_i}{\arr \siz_i \slk_i} \cdot (1-o(1)).
    \end{equation}
    Therefore, $\E\big[\wait\big] \geq  \max_{\indexmain \leq i \leq \njt} \frac{\sevmin \sysvar_i}{\arr \siz_i \slk_i} \cdot (1-o(1))$. This completes the proof.
\end{proof}

\begin{remark}
    The proof of the lower bound provides some intuitions for choosing the SNF policy. We consider the simple case where $\indexmain = \njt$. By \eqref{eq:lb-key-step}, the total queue length order-wise achieves the lower bound when the queue consists of jobs with the largest server needs, which suggests us to give low priorities to those jobs. This is in a similar spirit to SRPT, which leaves jobs with large remaining service times in the queue \citep[See, e.g., ][]{Har_13}.
\end{remark}

\section{Proof of Theorem~\ref{result:waiting-time-priority} (mean waiting time under SNF)}\label{sec:proof:priority}
To understand the behavior under the SNF policy, one key observation is that for each $i\in[\njt]$, the type~$i$ jobs are unaffected by type~$j$ jobs with $j > i$. As a result, we can learn about the original system by analyzing each $i$-th subsystem, which is obtained by removing all jobs of type $j$ with $j>i$.
Some subsystems are under relatively heavier traffic, or more precisely, subject to $\slk_i = \obrac{\frac{\sqrt{\sysvar_i}}{\log n}}$, 
while some subsystems are under lighter traffic. For those subsystems under relatively heavier traffic, Lemma~\ref{result:workload-lower} and Lemma~\ref{result:workload-upper} are enough for use; for those subsystems under lighter traffic, we sometimes 
use Lemma~\ref{lem:proof:total-server-need-expectation}, which is a more refined bound on the expected total server need $\E\left[\sumall \siz_i X_i\right]$, proved based on the two tail bounds in Lemma~\ref{lem:proof:workload-tail} and Lemma~\ref{lem:proof:total-server-need-tail}.
The proofs of Lemmas~\ref{lem:proof:workload-tail}, Lemma~\ref{lem:proof:total-server-need-tail}, and \ref{lem:proof:total-server-need-expectation} are in Appendix~\ref{sec:app:priority}.

\begin{restatable}{lemma}{workloadtail}\label{lem:proof:workload-tail}
Consider the multiserver-job system with $n$ servers satisfying $\Lm \leq \Lmratio \slk$ (Assumption~\ref{assump:lm-bound}). Letting $\xnom_i = \frac{\arr_i}{\sev_i}$, under any $\Lm$-work-conserving policy, the normalized work has the following tail bound: for any $\epsilon$ such that $0<\epsilon < \Lmratio$, there exists $\alpha_1  = \frac{2\epsilon\slk}{\sevmin}+\Theta\left(\frac{n\Lm}{\slk}\log n\right)$ and $\beta_1 = \Theta\left(\frac{n\Lm}{\slk}\right)$ such that for any $j\ge 0$,
\begin{equation}
    \Prob\left(\sumall \frac{\siz_i}{\sev_i} (X_i(\infty) - \xnom_i) \geq \alpha_1 + \beta_1 \cdot j\right) \leq  e^{-j}. \label{eq:proof:sub-hw:claim3-statement}
\end{equation}
\end{restatable}

\begin{restatable}{lemma}{tsntail}\label{lem:proof:total-server-need-tail}
    Consider the multiserver-job system with $n$ servers satisfying $\Lm \leq \Lmratio \slk$ (Assumption~\ref{assump:lm-bound}). Letting $\xnom_i = \frac{\arr_i}{\sev_i}$, under any $\Lm$-work-conserving policy, the total server need has the following tail bound:
    there exists $\alpha_2$ and $\beta_2$ with $\alpha_2 = \frac{\delta}{2} + \Theta(\frac{n\Lm}{\slk}\log n)$ and $\beta_2 = \Theta(\frac{n\Lm}{\slk})$ such that for any $j\ge 0$,
    \begin{equation}\label{eq:proof:sub-hw:claim4-statement}
        \Prob\left(\sumall \siz_i (X_i(\infty) - \xnom_i) \geq \alpha_2 + \beta_2 \cdot j\right) \leq  e^{-j}.
    \end{equation}
\end{restatable}

The proof technique of the two lemmas is using state space concentration successively: Lemma~\ref{lem:proof:workload-tail} relies on the state-space concentration implied by Lemma~\ref{lem:mminf-bounds}, while Lemma~\ref{lem:proof:total-server-need-tail} relies on the state-space concentration implied by Lemma~\ref{lem:mminf-bounds} and Lemma~\ref{lem:proof:workload-tail}. 

As a consequence of the first two lemmas, we can give a bound on the expectation of the total server need $\E[\sumall \siz_i X_i]$, in a different traffic regime than what is assumed in Assumption~\ref{assump:traffic}. This is useful for analyzing the dynamics of subsystems under SNF.
\begin{restatable}[Total server need upper bound under a lighter traffic]{lemma}{tsnexp}\label{lem:proof:total-server-need-expectation}
Consider the multi-server-job system with $n$ servers satisfying $\Lm \leq \Lmratio \slk$ (Assumption~\ref{assump:lm-bound}), and $\slk = \wbrac{\sqrt{n\Lm}\log n}$. Letting $\xnom_i = \frac{\arr_i}{\sev_i}$, under any $\Lm$-work-conserving policy, the expected total server need has the following upper bound:
\begin{equation}\label{eq:proof:tsn-expectation:statement}
    \E\left[\sumall \siz_i (X_i(\infty) - \xnom_i)\right] = \exp\left(-\Omega\left(   \frac{\slk^2}{n\Lm}\right)\right).
\end{equation}
\end{restatable}

As a quick digression, with Lemma~\ref{lem:proof:total-server-need-tail}, we can prove an upper bound on the queueing probability in Corollary~\ref{result:queueing-probability}.
We comment that this queueing probability bound significantly improves on the bound in \citet{WanXieHar_21_2} in a slightly more constrained traffic regime.
\begin{restatable}[Queueing probability]{corollary}{queueingprob}\label{result:queueing-probability}
    Consider the multiserver-job system with $n$ servers satisfying $\Lm \leq \Lmratio \slk$ (Assumption~\ref{assump:lm-bound}), and $\slk = \wbrac{\sqrt{n\Lm}\log n}$. Then the probability that a job arrival in steady-state experiencing queueing has the following upper bound:
    \begin{equation}\label{eq:qp-bound:statement}
        \Prob\left(\sumall \siz_i X_i(\infty) \geq n\right) = \exp\left(-\Omega\left(   \frac{\slk^2}{n\Lm}\right)\right).
    \end{equation}
\end{restatable}

Now we are ready to prove Theorem~\ref{result:waiting-time-priority}.
\begin{proof}[Proof of Theorem~\ref{result:waiting-time-priority}.]
Recall that by Little's Law, we have that $\E\left[ \wait \right] = \frac{1}{\arr}\sum_{i=1}^\njt \E[Q_i]$, and thus it suffices to bound $\E[Q_i]$'s.
We bound $\E[Q_i]$ separately for each $i\in [\njt]$ as follows:
\begin{equation}\label{eq:priority:individual-bdd-by-server-need}
    \E\left[Q_i\right] \leq \frac{1}{\siz_i} \E \left[\sum_{j=1}^i \siz_j Q_j \right] \leq \frac{1}{\siz_i} \E \left[\sum_{j=1}^i \siz_j \left(X_j - \xnom_j\right) \right],
\end{equation}
where we have used the fact that $Q_i$'s are non-negative, $Q_j = X_j - Z_j$ and $\E[Z_j] = \frac{\arr_j}{\sev_j} = \xnom_j$. Observe that under SNF, the dynamics of the first $i$ types of jobs are unaffected by the rest of the jobs. Therefore, the expectation
$\E \left[\sum_{j=1}^i \siz_j \left(X_j - \xnom_j\right) \right]$ can be viewed as the expected total server need in the $i$-th subsystem, which has slack capacity $\slk_i$, maximal server need $\siz_i$, and work variability $\sysvar_i$. We discuss the bound on $\E[Q_i]$ in three cases based on different relationships of $\slk_i$, $\siz_i$ and $\sysvar_i$.

\noindent \textbf{Case 1: $\slk_i = \obrac{\frac{\sqrt{\sysvar_i}}{\log n}}$}. 
Applying Lemma~\ref{result:workload-upper} to the $i$-th subsystem, we have
\begin{align*}
    \E \left[\sum_{j=1}^i \siz_j \left(X_j - \xnom_j\right) \right] 
    &\leq \sevmax \E \left[\sum_{j=1}^i \frac{\siz_j}{\sev_j} \left(X_j - \xnom_j\right) \right] \frac{\sevmax \sysvar_i}{\slk_i-\siz_i} \cdot (1+o(1)).
\end{align*}
Therefore, $\E[Q_i] \leq  \frac{\sevmax\sysvar_i}{\siz_i (\slk_i - \siz_i)} \cdot \left(1+o(1)\right)$.

\noindent \textbf{Case 2: $\slk_i = \Wbrac{\frac{\sqrt{\sysvar_i}}{\log n}}$ and $\siz_i = \wbrac{\frac{\slk \log n}{\sqrt{\sysvar}} \Lm}$}.
Applying Lemma~\ref{result:workload-upper} to the $i$-th subsystem, we have
\begin{align*}
    \E \left[\sum_{j=1}^i \siz_j \left(X_j - \xnom_j\right) \right] \leq \sevmax \E \left[\sum_{j=1}^i \frac{\siz_j}{\sev_j} \left(X_j - \xnom_j\right) \right]=\Obrac{\sqrt{\sysvar_i}\log n}.
\end{align*}
Combining the above equation with \eqref{eq:priority:individual-bdd-by-server-need} and the facts that $\siz_i = \wbrac{\frac{\slk \log n}{\sqrt{\sysvar}} \Lm}$ and $\sysvar_i \leq \sysvar$, we have
$
    \E[Q_i] \leq \Obrac{\frac{\sqrt{\sysvar_i}\log n}{\siz_i}} =\obrac{\frac{\sysvar}{\slk}\cdot\frac{1}{\Lm}}
$.

\noindent \textbf{Case 3: $\slk_i = \Wbrac{\frac{\sqrt{\sysvar_i}}{\log n}}$ and $\siz_i = \Obrac{\frac{\slk \log n}{\sqrt{\sysvar}} \Lm}$}. We first show that $\slk_i= \wbrac{\sqrt{n\siz_i}\log n}$. To see this, observe that $\slk_i = \Wbrac{\frac{\sqrt{\sysvar_i}}{\log n}}$ implies $i < \njt$, so $\slk_i \geq \frac{\arr_{\njt}\siz_{\njt}}{\sev_{\njt}} = n \load_{\njt}$. By Assumption~\ref{assump:commonness}, we have
\begin{equation}
    \slk_i \geq n \load_{\njt} = \wbrac{\sqrt{\frac{\slk \log n}{\sqrt{\sysvar} }\cdot n\Lm}\log n} = \wbrac{\sqrt{n\siz_i}\log n}.
\end{equation}

Therefore, we can apply Lemma~\ref{lem:proof:total-server-need-expectation} to the $i$-th subsystem to get:
\begin{align*}
    \E \left[\sum_{j=1}^i \siz_j \left(X_j - \xnom_j\right) \right] = \exp\left(-\Wbrac{\frac{\slk_i^2}{n\siz_i}}\right)\leq \exp\left(-\Wbrac{(\log n)^2}\right),
\end{align*}
where we have used $\slk_i= \wbrac{\sqrt{n\siz_i}\log n}$ in the last inequality. Therefore,
$
    \E[Q_i] \leq \exp\left(-\Wbrac{(\log n)^2}\right)
$, which decays faster than any polynomial in $n$.

Combining the three cases, we get
\begin{align*}
    \E\left[ \wait \right] = \frac{1}{\arr}\sum_{i=1}^\njt \E[Q_i] \leq \frac{1}{\arr}\sum_{i=\indexmain}^\njt \frac{\sevmax\sysvar_i}{\siz_i (\slk_i - \siz_i)} \cdot \left(1+o(1)\right),
\end{align*}
where we have used the fact that $\E\left[ \wait \right] =  \frac{1}{\arr}\sumall \E[Q_i]$. Here the summation is taken from $\indexmain$ to $\njt$ because Case~1 is equivalent to $\indexmain \leq i\leq \njt$. The second inequality is because the other two terms are of lower order compared with $\frac{\sevmax\sysvar_{\njt}}{\siz_{\njt}(\slk_{\njt}-\siz_{\njt})}$. 

It remains to show that  $\frac{1}{\arr}\sum_{i=\indexmain}^\njt \frac{\sevmax\sysvar_i}{\siz_i (\slk_i - \siz_i)} = \Thebrac{\max_{\indexmain \leq i \leq \njt}\frac{1}{\arr} \frac{\sysvar_i}{\siz_i \slk_i}}$. This is because $\njt$ and $\sevmax$ are independent of $n$ during scaling, and $\siz_i \leq \Lm \leq \Lmratio \slk \leq \Lmratio \slk_i$ for some $\Lmratio < 1$. This finishes the proof.

\end{proof}

\section{Simulation results}\label{sec:simulation}
\begin{figure}
    \centering
    \subfigure[Mean waiting time under the Parameter Set One. There are three job types, with service rates and server needs given by $\sev_1 = 0.25$, $\siz_1 = 1$; $\sev_2 = 0.5$, $\siz_2 = \lfloor\log_2 n\rfloor$; $\sev_3 =1 $, $\siz_3 =\lfloor \sqrt{n}\rfloor$; slack capacity $\slk = 2\lfloor \sqrt{n}\rfloor$; the arrival rates are chosen so that the loads satisfy $\load_1 = \load_2 = \load_3 = \frac{n-\slk}{3n}$.]{\includegraphics[width=0.45\textwidth]{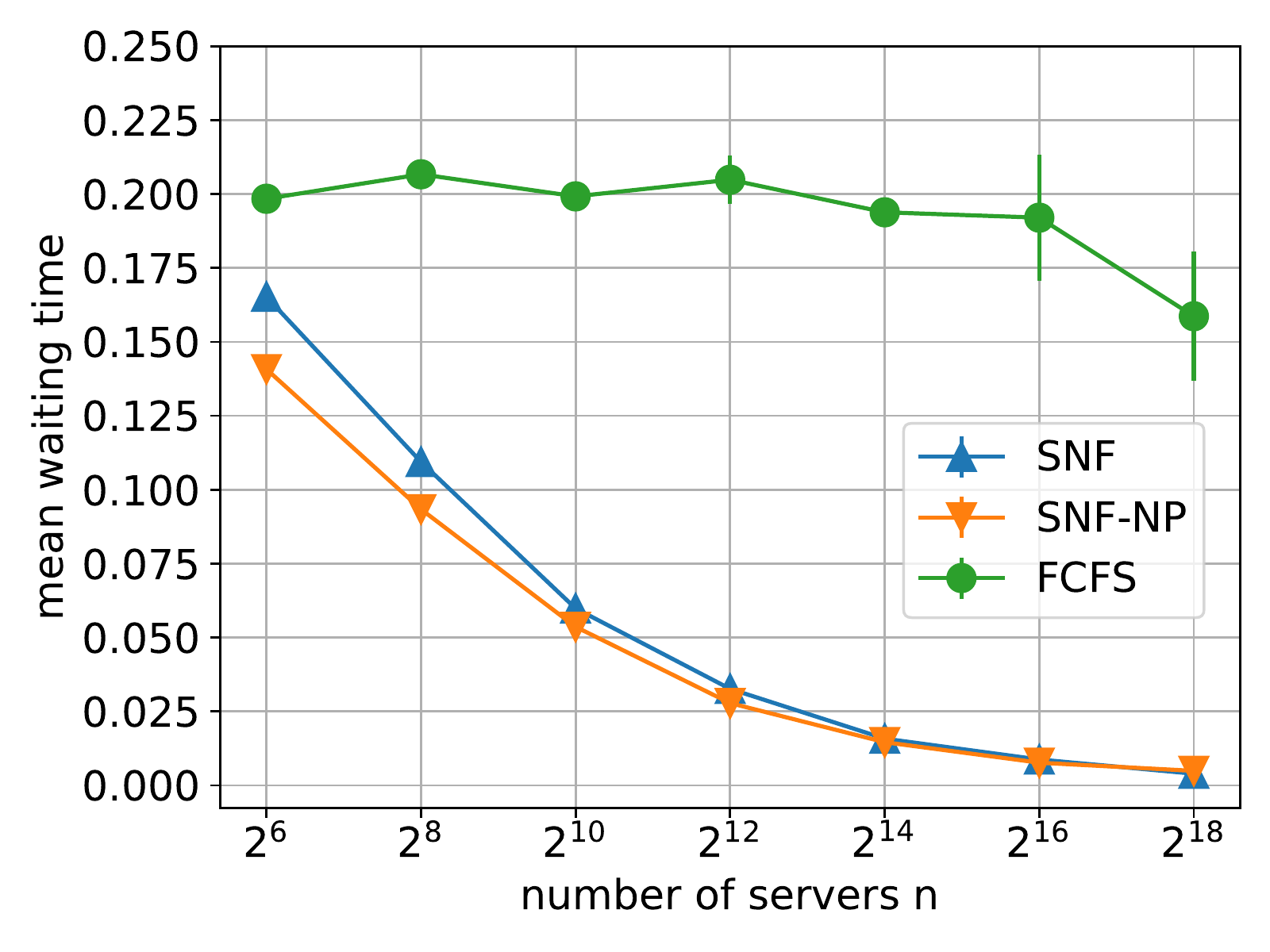}}
   \hspace{0.5cm}
    \subfigure[Mean waiting time under the Parameter Set Two. There are three job types, with the same service rates, server needs and slack capacity as the Parameter Set One. The arrival rates are chosen so that the loads satisfy $\load_1 = \load_2 = \frac{n - \slk - n^{0.7}}{2n}$, $\load_3 = n^{-0.3}$.]{\includegraphics[width=0.45\textwidth]{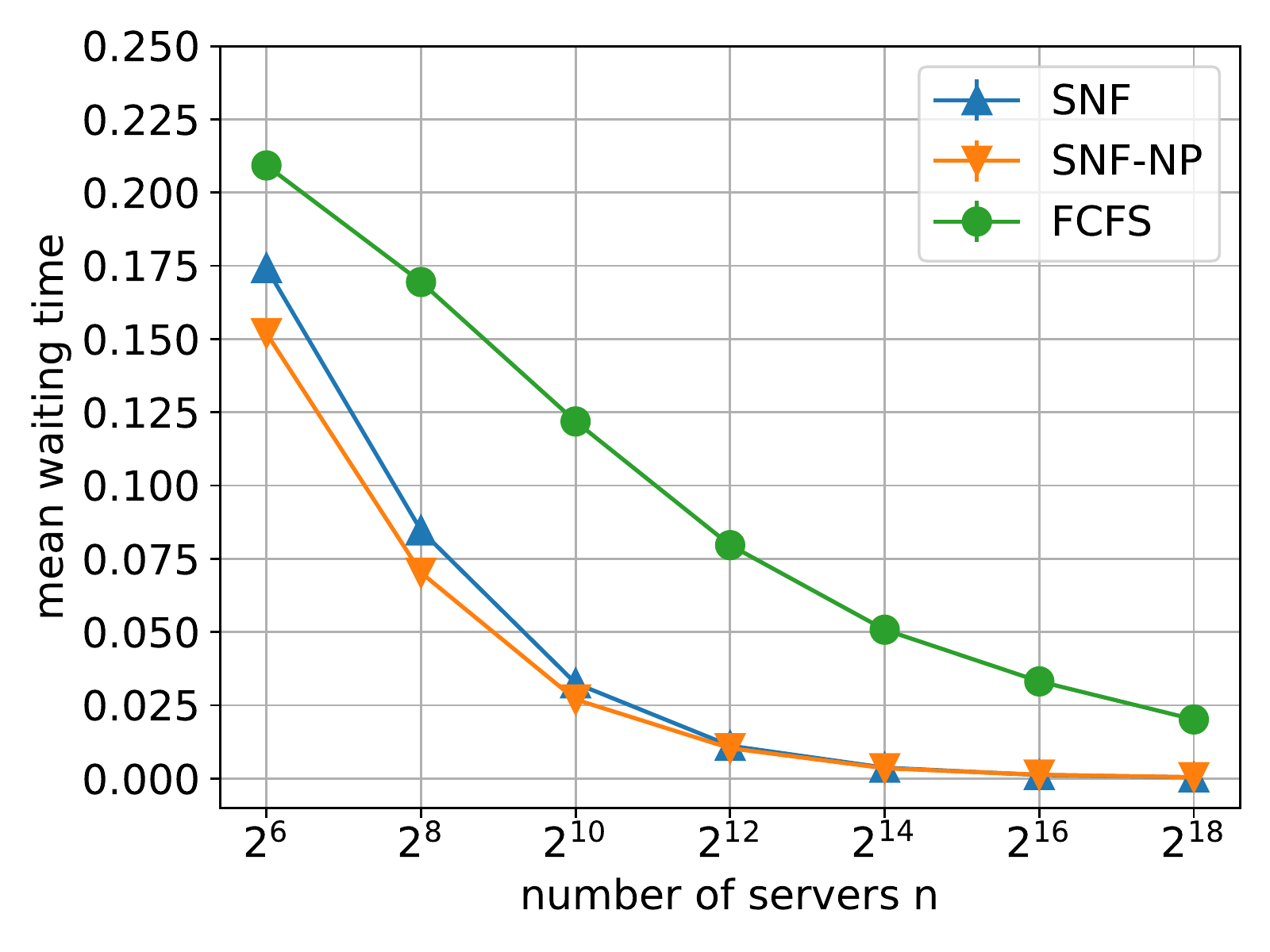}}
    \caption{The mean waiting times of FCFS, SNF, and SNF-NP under two sets of parameters. }
    \label{fig:mwt-exp}
\end{figure}

In this section, we perform simulation experiments to demonstrate the mean waiting time under FCFS, SNF and SNF-NP, where SNF-NP is the non-preemptive variant of SNF that serves a job with the smallest server need in the queue only when enough number of servers free up. 

We run the simulation experiment under two sets of parameters. The Parameter Set One satisfies all three assumptions, while the Parameter Set Two does not satisfy Assumption~\ref{assump:commonness}. The parameters are specified in the caption of Figure~\ref{fig:mwt-exp}.

We plot the mean waiting time against the number of servers $n$ under the three policies, as shown in Figure~\ref{fig:mwt-exp}. The parameter $n$ takes value in $\{2^6, 2^8, 2^{10}, 2^{12}, 2^{14}, 2^{16}, 2^{18}\}$. For each data point, we run a long enough trajectory to estimate the mean and the confidence interval. The confidence interval is estimated through batch means method \citep[see][]{AsmGly_07}, which divides the trajectory into $20$ batches, and calculates the variance of the means of each batch. It turns out that the confidence intervals of most data points in the plots are too small to be visible.

We can make the following observations from the experiments on both parameter sets. First, there is a large performance gap between FCFS and SNF for systems with finite number of servers $n$, which complements our asymptotic results. Note that although in the Parameter Set Two, the absolute different between FCFS and SNF seems to close up as $n$ gets large, their ratio is always greater than $3$ when $n \geq 2^{10}$, and it gets to as large as $40.0$ under Parameter Set One, and $54.4$ under Parameter Set Two. Second, SNF-NP performs comparably with SNF, and sometimes even performs slightly better.

\section{Conclusion and future work}\label{sec:conclusion}
In this paper, we have established order-wise sharp bounds on the mean waiting times of multiserver jobs under the FCFS policy and the SNF policy. We have also proved a lower bound of mean waiting time applicable to any policy. Those bounds imply the optimality of SNF and the strict sub-optimality of FCFS. Apart from the theoretical analysis, we have also demonstrated through simulations the performance improvement of SNF compared with FCFS in finite systems, and the fact that SNF-NP, which is the non-preemptive variant of SNF, has comparable performance with SNF.

There are several interesting directions for future work: (i) Derive a tighter bound on the mean waiting time under SNF when the commonness assumption is violated. (ii) Analyze the performance of SNF-NP. (iii) Relax the maximal server need assumption (Assumption~\ref{assump:lm-bound}), which may require new policy designs to ensure stability while keeping the mean waiting time small.

\bibliographystyle{informs2014}
\bibliography{refs-yige-v230411}

\appendix

\newpage
\section{Discussion on settings without Assumption~\ref{assump:commonness} (commonness assumption)}\label{sec:discussion}
Recall that Theorem~\ref{result:waiting-time-priority} is based on Assumption~\ref{assump:commonness} (commonness assumption). More generally, we have a bound on the mean waiting time of SNF without this assumption, although in that case the order-wise optimality of SNF is not guaranteed. We first define the necessary notation and then state the general bound. We omit the proof of the general bound since it is almost identical to the proof of Theorem~\ref{result:waiting-time-priority}.

Define $
\indexone \triangleq \min\left\{i\in[\njt] \;\middle|\; \slk_i = \Obrac{\sqrt{n\siz_i}\log n}\right\}$, and
$\indextwo \triangleq \indexmain = \min\left\{i\in[\njt]\;\middle|\; \slk_i=\obrac{\sqrt{\sysvar_i}/\log n}\right\}
$.
Because of Assumption~\ref{assump:traffic}, the two sets are non-empty so $\indexone$ and $\indextwo$ are well-defined. Because by definition $n\siz_i \geq \sevmin \sysvar_i$ and $\slk_i$'s are monotonically decreasing, we have $\indexone \leq \indextwo$. Moreover, $\sysvar_i$ and $\siz_i$ are monotonically increasing, so for any $\indextwo \leq i \leq \njt$, $\slk_i = \obrac{\sqrt{\sysvar_i}/\log n}$;
for any $\indexone \leq i < \indextwo$, $\slk_i = \Obrac{\sqrt{n\siz_i}\log n}\cap \Wbrac{\sqrt{\sysvar_i}/\log n}$; 
for any $1\leq i < \indexone$, $\slk_i = \wbrac{\sqrt{n\siz_i}\log n}$. 

\begin{theorem}[Waiting times of SNF without Assumption~\ref{assump:commonness}]
    Consider the multiserver-job system with $n$ servers satisfying Assumption~\ref{assump:traffic} and \ref{assump:lm-bound}. Under SNF, for each $i\in[\njt]$, the expected waiting time for type $i$ jobs and the mean waiting time satisfy:\\
    \begin{equation}\label{eq:discuss:individual-waiting-time-priority-case1}
        \E\big[\wait_i(\infty)\big]^\PP \leq \frac{1}{\arr_i} \frac{\sevmax\sysvar_i}{\siz_i (\slk_i - \siz_i)} \cdot \left(1+o(1)\right)\quad \text{ for $i\geq\indextwo$};
    \end{equation}
    \begin{equation}\label{eq:discuss:individual-waiting-time-priority-case2}
        \E\big[\wait_i(\infty)\big]^\PP \leq \Obrac{ \frac{1}{\arr_i} \frac{\sqrt{\sysvar_i}\log n}{\siz_i}}\quad \text{ for $\indexone\leq i<\indextwo$};
    \end{equation}
    \begin{equation}\label{eq:discuss:individual-waiting-time-priority-case3}
        \E\big[\wait_i(\infty)\big]^\PP \leq \exp\left(-\Wbrac{\frac{\slk_i^2}{n\siz_i}}\right)\quad \text{ for $1\leq i<\indexone$}.
    \end{equation}
    Thus, the mean waiting time satisfies
    \begin{equation}\label{eq:discuss:mean-waiting-time-priority}
        \mwt^\PP \leq \frac{1}{\arr} \sum_{i=\indextwo}^\njt \frac{\sevmax\sysvar_i}{\siz_i (\slk_i - \siz_i)} \cdot \left(1+o(1)\right) + \sum_{k=\indexone}^{\indextwo-1} \Obrac{ \frac{1}{\arr} \frac{\sqrt{\sysvar_k}\log n}{\siz_k}}.
    \end{equation}
\end{theorem}
The proof of this theorem is almost identical to the proof of  Theorem~\ref{result:waiting-time-priority}. In fact, we have already derived the bounds in \eqref{eq:discuss:individual-waiting-time-priority-case1}, \eqref{eq:discuss:individual-waiting-time-priority-case2} and \eqref{eq:discuss:individual-waiting-time-priority-case3} in the proof of Theorem~\ref{result:waiting-time-priority} in Section~\ref{sec:proof:priority} without using Assumption~\ref{assump:commonness}.

Note that when we have the commonness assumption, the second term on the right of \eqref{eq:discuss:mean-waiting-time-priority} has a strictly smaller order than the first term (see the proof of Theorem~\ref{result:waiting-time-priority} in Section~\ref{sec:proof:priority}), so the upper bound orderwise matches the lower bound. Without the commonness assumption, it is unclear whether SNF achieves the optimal order. We leave a tighter bound on the mean waiting time under SNF for future work.

\section{Conditions for applying drift method}\label{sec:app:drift-cond}
\subsection{Drift method in general markov chain}\label{sec:app:drift-cond-general}
In this subsection, we consider a continuous-time Markov chain $\{\ve{S}(t)\}_{t\geq 0}$ on a countable state space $\mathcal{S}$ with generator $G$. Let $r_{ss'}$ denote the transition rate from state $\ve{s}$ to $\ve{s}'$, and let $r(\ve{s}) = -r_{ss}\triangleq\sum_{\ve{s}'\in\mathcal{S},\ve{s}'\neq \ve{s}} r_{ss'}.$
Then for any function $f\colon \mathcal{S}\to\mathbb{R}$, the drift of $f$ is equal to:
\begin{equation}
G f(\ve{s}) = \sum_{\ve{s}'\in \mathcal{S}} r_{ss'} \left(f(\ve{s}') - f(\ve{s})\right).
\end{equation}
We assume that $\{\ve{S}(t)\}_{t\geq 0}$ has a stationary distribution and let $\ve{S}(\infty)$ denote a random element that follows the stationary distribution.

We give a justification of the frequently used condition:
\begin{equation}
    \E[G f(\ve{S}(\infty))] = 0. \label{eq:bar}
\end{equation}
Lemma~\ref{lem:bar-master-condition} below, which is a restatement of the Proposition~3 of \citet{GlyZee_08}, gives rigorous conditions under which the relation in \eqref{eq:bar} holds.

\begin{lemma}\label{lem:bar-master-condition}
     Consider the Markov chain $\{\ve{S}(t)\}_{t\geq 0}$ and a function $f\colon\mathcal{S}\to\mathbb{R}$.
     If for all $\ve{s}\in\mathcal{S}$,
     \begin{equation*}          \E[r(\ve{S}(\infty))\cdot\abs{f(\ve{S}(\infty))}] < \infty,
     \end{equation*}
     then the relation $\E[G f(\ve{S}(\infty))] = 0$ holds.
\end{lemma}

Lemma~\ref{lem:bar-master-condition} implies that if
\begin{gather*}
\sup_{\ve{s}\in\mathcal{S}} r(\ve{s}) < \infty,\\
\E[\abs{f(\ve{S}(\infty)}] < \infty,
\end{gather*}
then the relation $\E[G f(\ve{S}(\infty))] = 0$ holds.

Lemma~\ref{lem:hajek} below, which is a continuous-time analogue of a result in \citet{haj_82}, provides some functions $f$ that satisfy $\E[\abs{f(\ve{S}(\infty))}] < \infty$.
\begin{lemma}\label{lem:hajek}
    Consider the Markov chain $\{\ve{S}(t)\}_{t\geq 0}$ and let $V:\mathcal{S}\to\R_+$ be a Lyapunov function. Assume that
    \begin{align*}
    &\sup_{\ve{s}\in\mathcal{S}} r(\ve{s}) < \infty,\\
    &\sup_{\ve{s},\ve{s}'\in\mathcal{S}: r_{ss'} > 0} \abs{V(\ve{s}') - V(\ve{s})} < \infty,
    \end{align*}
    and that there exists $\gamma>0$ and $B\geq 0$, such that when $V(\ve{s}) \geq B$,
    \begin{equation*}
        G V(\ve{s}) \leq -\gamma.
    \end{equation*}
    Then for any positive integer $m$, 
    \begin{equation*}
        \E[V(\ve{S}(\infty))^m] < \infty.
    \end{equation*}
\end{lemma}

\subsection{Multiserver-job system}\label{sec:app:drift-cond-msj}
We return to the multiserver-job system. For $f\colon\R^\njt \to \R$, we want to find conditions under which $E[G f(\ve{X},\ve{Z})] = 0$. Recall that
\begin{equation}
    G f(\ve{x}, \ve{z}) = \sumall \arr_i \left(f(\ve{x}+ \ve{e}_i) - f(\ve{x})\right)+ \sumall \sev_i z_i \left(f(\ve{x}-\ve{e}_i) - f(\ve{x})\right),
\end{equation}
where $\ve{x}, \ve{z}$ are possible realizations of $\ve{X}, \ve{Z}$, $\ve{e}_i \in \R^\njt$ is the vector whose $i$th entry is $1$ and all the other entries are $0$. Therefore, the transition rate of any multiserver-job system is bounded by $\sevmax n$. By Lemma~\ref{lem:bar-master-condition}, it suffices to have $\E[\abs{f(\ve{X})}] < \infty$.

This lemma implies that for any test function $f$ with polynomial increasing rate, the relation $\E[Gf(\ve{X}, \ve{Z})] = 0$ holds.
\begin{lemma}\label{lem:polynomial-zero-drift}
    Consider the multiserver-job system with $n$ servers under a  $\wst$-work-conserving policy with $\wst \leq \Lmratio \slk$, where $\Lmratio \in (0,1)$ is the parameter used in Assumption~\ref{assump:lm-bound}. Then for any positive integer $m$, we have
    \begin{equation}
        \E\left[\left(\sumall \frac{\siz_i}{\sev_i}X_i(\infty)\right)^m\right] < \infty.
    \end{equation}
\end{lemma}

\begin{proof}[Proof of Lemma~\ref{lem:polynomial-zero-drift}.]
    We prove this lemma by applying Lemma~\ref{lem:hajek}. Let $V(\ve{x}) = \sumall \frac{\siz_i}{\sev_i} x_i$.  We first show that the drift $GV(\ve{x},\ve{z})$ is negative when $V(\ve{x})$ is larger than a threshold. Observe that this drift can be bounded as
    \begin{align*}
        G V(\ve{x},\ve{z}) &= \sumall \left(\frac{\arr_i\siz_i}{\sev_i} - \siz_i z_i\right)\\
        &\leq \sumall \frac{\arr_i\siz_i}{\sev_i} - \min\left(\sumall \siz_i x_i, n-\wst\right)\\
        &\leq -\min\left(\sumall \siz_i x_i - n + \slk, \slk - \wst\right),
    \end{align*}
    where the second inequality is due to the use of $\wst$-work-conserving policy. For any state such that $V(\ve{x}) > \frac{n}{\sevmin} + 1$ with $\sevmin = \min_{i\in[\njt]} \sev_i$, we have 
    \[
        \sumall \siz_i x_i \geq \sevmin \sumall \frac{\siz_i}{\sev_i} x_i = \sevmin V(\ve{x}) > n,
    \]
    so $G V(\ve{x},\ve{z}) < -\slk + \wst \leq -\Lmratio \slk < 0$. Moreover, any transition in the system at most changes $V(\ve{x})$ by $\max_i \frac{\siz_i}{\sev_i}$, which is also a finite number. Applying Lemma~\ref{lem:hajek}, we get $\E\left[\left(\sumall \frac{\siz_i}{\sev_i} X_i\right)^m\right] < \infty$ for any positive integer $m$.
\end{proof}

\section{State-space concentration by coupling with infinite server systems}\label{sec:mminf-proof}
In this section, we prove Lemma~\ref{lem:mminf-bounds}, the state-space concentration result that we have stated in Section~\ref{sec:workload-bounds}. 

We start by defining a sequence of infinite-server systems. Consider a sequence of systems where each system has an infinite number of servers.  For the $n$th system in the sequence, we let it have the same job characteristics as the original $n$-server system; i.e., there are $\njt$ job types and job type $i\in[\njt]$ has arrival rate $\arr_i$, server need $\siz_i$, and service rate $\sev_i$.  We refer to this system as the \emph{$n$th infinite-server system}. The $n$th infinite-server system serves as a lower bound system for the original $n$-server system through the coupling defined below.

Let the $n$th infinite-server system have the same arrival sequence as the original $n$-server system; i.e., whenever the $n$-server system has a job arrival, let the infinite-server system have a job arrival with the same server need and service time.  Since every job in the infinite-server system enters service immediately upon arrival, the job leaves the infinite-server system no later than the corresponding job does in the original system. Therefore, the infinite-server system always has no more jobs than the $n$-server system.  Specifically, let $\sysinf{X}_i(t)$ be the number of type~$i$ jobs in the infinite-server system at time $t$.  Recall that $X_i(t)$ denotes the number of type~$i$ jobs in the original $n$-server system.  Then by the construction of the coupling, $\sysinf{X}_i(t) \leq X_i(t)$ almost surely for all $t\geq 0$ and $i\in [\njt]$.  Therefore, in steady-state, $\sysinf{X}_i(\infty)$ is stochastically smaller or equal to $X_i(\infty)$, i.e., $\Prob(\sysinf{X}_i(\infty)\geq x) \leq \Prob(X_i(\infty) \geq x)$ for any $x\in \R$. We denote this relationship by $\sysinf{X}_i(\infty) \leq_{st} X_i(\infty)$.

Next, we restate Lemma~\ref{lem:mminf-bounds} and give it a proof.

\mminf*

\begin{proof}[Proof of Lemma~\ref{lem:mminf-bounds}.]
Consider the infinite-server system defined at the beginning of this section. Based on the coupling argument, we have $\sysinf{X}_i(\infty) \leq_{st} X_i(\infty)$ for any $i\in[\njt]$. We define
\[
\ycinf =\sumall c_i \siz_i\left(\sysinf{X}_i(\infty)-\frac{\arr_i}{\sev_i}\right).
\]
Because $c_i \siz_i\geq 0$, we have 
\begin{equation}\label{eq:app:proof:mminf-yc-coupling}
   \ycinf \leq_{st} \yc.
\end{equation}
Therefore, we can first prove the various lower bounds for $\ycinf$, and then apply \eqref{eq:app:proof:mminf-yc-coupling}. 

We first prove the bound in (a). Since $\sysinf{X}_i(\infty)$ is the number of jobs in steady-state in a M/M/$\infty$ system with arrival rate $\arr_i$ and service rate $\sev_i$, $\sysinf{X}_i(\infty)$ follows Poisson$\left(\frac{\arr_i}{\sev_i}\right)$ distribution by classic results \citep[see, e.g.,][p.249]{BolGreMeeTri_06}. By the independence of $\sysinf{X}_i(\infty)$'s, we have 
\begin{equation*}
    \E[e^{-t \ycinf}] =
    \exp\left( \sumall \frac{\arr_i}{\sev_i} (e^{-c_i \siz_i t} - 1 + c_i \siz_i t)\right).   
\end{equation*}

We first prove the bound in (a). By Markov inequality, for any $t\ge 0$,
\begin{align*}
    \Prob\left(\ycinf \leq - K \right) &= \Prob\left( e^{-t \ycinf} \geq e^{t K}\right)\leq \frac{\E[e^{-t\ycinf}]}{e^{t K}}.
\end{align*}
To upper-bound $\E[e^{-t\ycinf}]$, observe that $e^{-s} \leq 1-s +\frac{s^2}{2}$ for any $s \geq 0$. Then for any $t \geq 0$,
\begin{align*}
\E[e^{-t \ycinf}] &=
    \exp\left( \sumall \frac{\arr_i}{\sev_i} \left(e^{-c_i \siz_i t} - 1 + c_i \siz_i t\right)\right) 
    \le \exp\left( \frac{1}{2} \sumall \frac{\arr_i}{\sev_i} \left(c_i \siz_i t\right)^2\right)\\
    &= \exp\left(\frac{1}{2}\sumall c_i^2\sev_i\cdot \frac{\arr_i\siz_i^2}{\sev_i^2}t^2\right) 
    \leq \exp\left(\frac{1}{2}\cmax^2 \sevmax \sysvar t^2\right),
\end{align*}
where we have used the facts that $c_i^2\le \cmax^2$ and $\sev_i \leq \sevmax$ for any $i\in[\njt]$. Therefore,
\begin{align*}
    \Prob\left(\ycinf \leq - K \right)
    &\leq \exp\left(\frac{1}{2} \cmax^2 \sevmax \sysvar t^2 - K t\right).
\end{align*}
Take $t = \frac{K}{\cmax^2 \sevmax \sysvar}$, we get
\begin{equation}
    \Prob\left(\yc \leq - K \right) \leq \Prob\left(\ycinf \leq - K \right) \leq \exp\left(-\frac{K^2}{2\cmax^2 \sevmax \sysvar}\right),
\end{equation}
where we have used $\ycinf \leq_{st} \yc$ in the first inequality. This proves (a).

Next we prove the bound in (b). Let $\alpha$ and $\beta$ be any two non-negative numbers such that $\alpha\beta \geq \cmax^2\sevmax \sysvar$ and let $j\ge 0$. By the bound in (a), setting $K = \alpha + \beta j$ gives
\begin{align*}
    \Prob\left(\ycinf \leq - \alpha - \beta j \right)
    \leq \exp\left(-\frac{\alpha^2 + 2\alpha\beta j + \beta^2 j^2}{2\cmax^2 \sevmax \sysvar}\right)
    \leq \exp\left(-\frac{\alpha\beta j}{\cmax^2\sevmax \sysvar}\right)
    \le e^{-j}.
\end{align*}
Because $\ycinf \leq_{st} \yc$, 
\[
    \Prob\left(\yc \leq - \alpha - \beta j \right) \leq e^{-j}.
\]
This proves (b).

To prove the bound on $\E\left[\left(\ycinf\right)^-\right]$ in (c), observe that
\begin{equation}
    \E\left[\left(\ycinf\right)^-\right] \leq \E\left[\abs{\ycinf}\right] \leq \sqrt{\E\left[\left(\ycinf\right)^2\right]},
\end{equation}
\begin{align*}
    \E\left[\left(\ycinf\right)^2\right] &= \sumall c_i^2 \siz_i^2\E\left[\left(\sysinf{X}_i(\infty)-\frac{\arr_i}{\sev_i}\right)^2\right]
    =\sumall c_i^2 \siz_i^2 \frac{\arr_i}{\sev_i}
    \leq \cmax^2 \sevmax \sysvar,
\end{align*}
where we have used the fact that $\sysinf{X}_i(\infty)$ follows Poisson$\left(\frac{\arr_i}{\sev_i}\right)$ distribution. Because $s\mapsto s^-$ is a non-increasing function for real number $s$, we have 
\begin{equation}
    \E\left[\yc^-\right]\leq \E\left[\left(\ycinf\right)^-\right] \leq \sqrt{\cmax^2 \sevmax \sysvar}.
\end{equation}
This finishes the proof of (c).
\end{proof}

\section{Proofs of the lemmas for Theorem~\ref{result:waiting-time-priority} (mean waiting time under SNF)}\label{sec:app:priority}
\subsection{A tail bound in general markov chains based on Lyapunov condition}\label{sec:app:tail-bounds-general}
In this subsection, we prove a tail bound in general markov chains based on a certain Lyapunov drift condition. Consider a continuous-time Markov chain $\{\ve{S}(t)\}_{t\geq 0}$ on a countable state space $\mathcal{S}$ with generator $G$. Let $r_{ss'}$ denote the transition rate from state $\ve{s}$ to $\ve{s}'$, and let $r(\ve{s}) = -r_{ss}\triangleq \sum_{\ve{s}'\in\mathcal{S},\ve{s}'\neq \ve{s}} r_{ss'}.$
For any function $f\colon \mathcal{S}\to\mathbb{R}$, the drift of $f$ is equal to
\begin{equation}
G f(\ve{s}) = \sum_{\ve{s}'\in \mathcal{S}} r_{ss'} \left(f(\ve{s}') - f(\ve{s})\right).
\end{equation}
We assume that $\{\ve{S}(t)\}_{t\geq 0}$ has a stationary distribution and let $\ve{S}(\infty)$ denote a random element that follows the stationary distribution.

For a Lyapunov function $V: \mathcal{S}\rightarrow \R_+$, Lemma~\ref{lem:lyapunov-upper-bound} below gives a tail bound on $V(\ve{S}(\infty))$ based on drift analysis. It is a straightforward consequence of Lemma~A.1 in \citet{WenWan_21_2}, which slightly generalizes the well-known Lyapunov-based tail bounds \citep[see, e.g.,][]{BerGamTsi_01,WanMagSri_17_2}.

\begin{lemma}[Lyapunov Upper Bound]\label{lem:lyapunov-upper-bound}
Consider the Markov chain $\{\ve{S}(t)\}_{t\geq 0}$ and
let $V: \mathcal{S}\rightarrow \R_+$ be a Lyapunov function such that $\E[V(\ve{S}(\infty))] < \infty$. Assume that
\begin{align*}
&\sup_{\ve{s}\in\mathcal{S}} r(\ve{s}) < \infty,\\
&\vmx \triangleq \sup_{\ve{s},\ve{s}'\in\mathcal{S}: r_{ss'} > 0} \abs{V(\ve{s}') - V(\ve{s})} < \infty,\\
&\fmx \triangleq \sup_{\ve{s}} \sum_{\ve{s}':V(\ve{s}')>V(\ve{s})} r_{ss'}\left(V(\ve{s}') - V(\ve{s})\right) <\infty.
\end{align*}
If there exists $B\geq 0$, $\gamma > 0,$ $\xi > 0$ and a subset $\mathcal{E}\subseteq \mathcal{S}$ such that for all $\ve{s}$ with $V(\ve{s}) > B+\vmx$, 
\begin{equation*}
    G V(\ve{s}) \leq 
    \begin{cases}
    - \gamma, \quad &\forall \ve{s}\in \mathcal{E},\\
    \xi, \quad &\forall \ve{s}\in \mathcal{S}\backslash\mathcal{E},
    \end{cases}
\end{equation*}
then for all $j\geq 0$,
\begin{equation}\label{eq:lyapunov-upper:conclusion}
    \Prob\left(V(\ve{S}(\infty)) \ge B +2\vmx\cdot\left(\left(\frac{ \fmx}{\gamma}+2\right)j+1\right)\right) 
    \leq e^{-j} + \left(\frac{\xi}{\gamma}+1\right)\Prob(\ve{S}(\infty) \in \mathcal{E}^c).
\end{equation}
\end{lemma}

\begin{proof}[Proof of Lemma~\ref{lem:lyapunov-upper-bound}.]
According to Lemma~1 of \citet{BerGamTsi_01}, for any integer $j'\geq 0$,
\begin{equation*}
    \begin{aligned}
        \Prob\left(V(S(\infty)) > B +  2\vmx\cdot j'\right) &\leq \left(\frac{\fmx}{\fmx+\gamma}\right)^{j'} +\left(\frac{\xi}{\gamma}+1\right)\Prob(S(\infty) \in \mathcal{E}^c). \\
        &\leq \exp\left(-\frac{\gamma j'}{\fmx + \gamma}\right) + \left(\frac{\xi}{\gamma}+1\right)\Prob(S(\infty) \in \mathcal{E}^c)
    \end{aligned}
\end{equation*}
For any $j \geq 0$, we set $j' = \lceil( \frac{\fmx}{\gamma} +1)j \rceil$, which is an integer less than or equal to the value $\left(\left(\frac{ \fmx}{\gamma}+2\right)j+1\right)$.
Then \eqref{eq:lyapunov-upper:conclusion} follows.
\end{proof}

\subsection{Tail bounds in the multiserver-job system}\label{sec:app:tail-bounds-msj}
With the general tool Lemma~\ref{lem:lyapunov-upper-bound} ready, we can now start proving Lemma~\ref{lem:proof:workload-tail}, Lemma~\ref{lem:proof:total-server-need-tail} and Lemma~\ref{lem:proof:total-server-need-expectation}, which are needed for proving Theorem~\ref{result:waiting-time-priority}. As a by-product, we will also prove the bound on queueing probability stated in Corollary~\ref{result:queueing-probability}.

\workloadtail*

\tsntail*


\begin{proof}[Proof of Lemma~\ref{lem:proof:workload-tail}.]
The idea of the proof is applying Lemma~\ref{lem:lyapunov-upper-bound} to the Lyapunov function $V_1(\ve{x}) = \sumall \frac{\siz_i }{\sev_i} x_i$, where $\ve{x} \in \R^\njt$ denotes possible realizations of $\ve{X}$. To check the conditions of Lemma~\ref{lem:lyapunov-upper-bound}, first recall that the drift of $V_1(\ve{x})$ is
\begin{equation}
    \begin{aligned}
      G V_1(\ve{x}, \ve{z}) = \sumall \arr_i \left(V_1(\ve{x}+\ve{e}_i) - V_1(\ve{x})\right) 
     + \sumall  \sev_i z_i \left(V_1(\ve{x}-\ve{e}_i) - V_1(\ve{x})\right).
    \end{aligned}\label{eq:proof:gv1-full}\\
\end{equation}
From the formula of the drift, we can immediately see that the two parameters $\vmx$ and $\fmx$ in Lemma~\ref{lem:lyapunov-upper-bound} are $\vmx = \max_{i \in [\njt]} \frac{\siz_i}{\sev_i}$ and $\fmx=\sumall \frac{\arr_i\siz_i}{\sev_i}$. The two parameters satisfy $\vmx \leq \frac{\Lm}{\sevmin}$ and $\fmx = n-\slk \leq n$.

We spend the rest of the proof showing that when $V_1(\ve{x})$ is larger than a threshold,
\begin{equation*}
    \sumall \frac{\siz_i}{\sev_i}x_i = V_1(\ve{x}) > \sumall \frac{\siz_i}{\sev_i}\xnom_i + B_1,
\end{equation*}
for some $B_1 \geq 0$, there exists $\gamma, \xi > 0$, such that with high probability we have $G V_1(\ve{x}, \ve{z})<-\gamma$ and in the worst case we have $G V_1(\ve{x},\ve{z}) < \xi$. We prove this by finding a suitable state-space concentration. To see what we need, we write out the form of $G V_1$ explicitly:
\begin{align}
    G V_1(\ve{x},\ve{z}) &= \sum_{i=1}^\njt \left(\frac{\arr_i\siz_i}{\sev_i} - \siz_i z_i\right) \label{eq:proof:gv1:temp0} \\
    &\leq \sum_{i=1}^\njt \frac{\arr_i\siz_i}{\sev_i} - \min\left(\sum_{i=1}^\njt \siz_i x_i, n - \Lm\right) \nonumber \\
    &= -\min\left(\sumall \siz_i (x_i-\xnom_i), \slk - \Lm\right) \nonumber\\
    &\leq -\min\left(\sumall \siz_i (x_i-\xnom_i), \Lmratio\slk\right), \label{eq:proof:gv1:temp}
\end{align}
where the first inequality is because of the $\Lm$-work-conserving condition, and the last inequality used the fact that $\Lm \leq \Lmratio\slk$. The last expression suggests that we need a high probability lower bound on $\sumall \siz_i X_i$. 

We apply Lemma~\ref{lem:mminf-bounds} to construct a subset $\mathcal{E}_{3}(j_1) \subseteq \R^\njt$ such that $\ve{X} \in \mathcal{E}_{3}(j_1)$ with high probability, where $j_1$ is a non-negative number. We set $c_i = \frac{1}{\sevmin}- \frac{1}{\sev_i}$, $\alpha = \frac{\epsilon \slk}{\sevmin}$, $\beta = \frac{\sevmax n\Lm}{\sevmin^2 \epsilon\slk}$, $j = 2\log n + j_1$ in Lemma~\ref{lem:mminf-bounds}. It can be verified that $\alpha\beta \geq \frac{\sevmax n \Lm}{\sevmin^3} \geq \cmax^2 \sevmax \sysvar$. Then we have
\begin{equation}
    \Prob\left(\sumall\left(\tfrac{1}{\sevmin} - \tfrac{1}{\sev_i} \right)\siz_i(X_i - \xnom_i) \leq -\alpha - \beta j_1 \right) \leq e^{-j_1}.
\end{equation}
Letting $\alpha_{3} = \alpha + 2 \beta \log n$, $\beta_{3} = \beta$, the above inequality can also be written as
\begin{equation}
    \Prob\left(\sumall\left(\tfrac{1}{\sevmin} - \tfrac{1}{\sev_i} \right)\siz_i(X_i - \xnom_i) \leq -\alpha_{3} - \beta_{3} j_1  \right) \leq \frac{1}{n^2} e^{-j_1}.
\end{equation}
Note that this choice of $\alpha_{3}$ and $\beta_{3}$ satisfies $\beta_{3} = \Theta\left(\frac{n\Lm}{ \slk}\right)$ and $\alpha_{3} = \frac{\epsilon \slk}{\sevmin} + \Theta\left(\frac{n\Lm}{ \slk}\log(n)\right)$. If we define the set $\mathcal{E}_{3}(j_1)$ as
\[
\mathcal{E}_{3}(j_1) = \left\{\ve{x}\in\R^\njt \;\middle|\; \sumall \left(\tfrac{1}{\sevmin} - \tfrac{1}{\sev_i} \right)\siz_i(x_i - \xnom_i) > -\alpha_{3} - \beta_{3} j_1 \right\},
\]
then the above inequality implies that $\Prob\left(\mathcal{E}_{3}(j_1)\right) \geq 1-\frac{1}{n^2}e^{-j_1}$. 

Next, we divide the state space into two parts based on $\mathcal{E}_{3}(j_1)$ and bound $G V(\ve{x},\ve{z})$ on each part separately. When $V_1(\ve{x}) > \sumall \frac{\siz_i}{\sev_i}\xnom_i + B_1$ and $\ve{x}\in\mathcal{E}_{3}(j_1)$, we have
\begin{equation*}
    \sumall \frac{\siz_i}{\sev_i}(x_i - \xnom_i) > B_1,
\end{equation*}
\begin{equation*}
\sumall \left(\tfrac{1}{\sevmin} - \tfrac{1}{\sev_i} \right)\siz_i(x_i - \xnom_i) > -\alpha_{3} - \beta_{3} j_1.
\end{equation*}
Adding up the above two inequality and choosing $B_1 = \alpha_{3} + \beta_{3} j_1+\frac{\epsilon \slk}{\sevmin}$ yield
\begin{equation}
    \sum_{i=1}^\njt \siz_i(x_i - \xnom_i) \geq \epsilon\slk.
\end{equation} 
By (\ref{eq:proof:gv1:temp}) and the fact that $\epsilon < \Lmratio$, we conclude that $G V(\ve{x},\ve{z}) \leq - \epsilon \slk$.

On the other hand, when $V_1(\ve{x}) > \sumall \frac{\siz_i}{\sev_i}\xnom_i + B_1$ and $\ve{x}\notin \mathcal{E}_{3}(j_1)$, it follows from \eqref{eq:proof:gv1:temp0} that $GV(\ve{x},\ve{z}) \leq n$. 

Now we can apply Lemma~\ref{lem:lyapunov-upper-bound} with $B=\sumall \frac{\siz_i}{\sev_i}\xnom_i + B_1$, $\gamma = -\epsilon \slk$, $\xi = n$, $\vmx = \max_{i\in[\njt]}\frac{\siz_i}{\sev_i}$, $\fmx = \sumall \frac{\arr_i \siz_i}{\sev_i}$ and $\mathcal{E} = \mathcal{E}_{3}(j_1)$ to get
\begin{equation}
    \begin{aligned}
        &\mspace{23mu}\Prob\Bigg(\sumall \frac{\siz_i}{\sev_i} (X_i - \xnom_i) \geq \alpha_{3}
        +  \beta_{3} j_1 + \frac{\epsilon\slk}{\sevmin}+ 2\vmx\left(\left(\frac{\fmx}{\epsilon\slk}+2\right)j_2+1\right)\Bigg) \\
        &\leq e^{-j_2} + \left(\frac{n}{\epsilon\slk}+1\right)\frac{1}{n^2}e^{-j_1}. 
    \end{aligned}
\end{equation}
which holds for any $j_1, j_2 \geq 0$. This inequality can be further simplified to
\begin{equation}
    \begin{aligned}
        &\mspace{23mu}\Prob\Bigg(\sumall \frac{\siz_i}{\sev_i} (X_i - \xnom_i) \geq \alpha_{3}
        +  \beta_{3}j_1 + \frac{\epsilon\slk}{\sevmin} + 2\frac{\Lm}{\sevmin}\left(\frac{(2\Lmratio+1)n}{\epsilon\slk}j_2+1\right)\Bigg)\leq e^{-j_2} + e^{-j_1}. 
    \end{aligned}
\end{equation}
where we have used the facts that $\vmx \leq  \frac{\Lm}{\sevmin}$, $\fmx \leq n$, $n>\slk$, $\epsilon < \Lmratio$, and $\left(\frac{n}{\epsilon\slk}+1\right)\frac{1}{n^2} \leq 1$ which holds for $n$ large enough. 
Replacing both of $j_1$ and $j_2$ with $j+\log(2)$, and letting $\beta_1 = \beta_{3} + \frac{(4\Lmratio+2) n \Lm}{\sevmin \epsilon \slk}$, $\alpha_1 = \alpha_{3} + \frac{\epsilon \slk}{\sevmin} + \frac{2\Lm}{\sevmin} + \beta_1 \log(2)$, we get the equation in the lemma statement. Note that this choice of $\alpha_1$ and $\beta_1$ implies that $\beta_1 =\Theta\left(\frac{n\Lm}{\slk}\right)$ and $\alpha_1 = \frac{2\epsilon\slk}{\sevmin}+\Theta\left(\frac{n\Lm}{\slk}\log(n)\right)$. This proves Lemma~\ref{lem:proof:workload-tail}.
\end{proof}

\begin{proof}[Proof of Lemma~\ref{lem:proof:total-server-need-tail}.]
Again we try to prove the result by applying Lemma~\ref{lem:lyapunov-upper-bound} to the Lyapunov function $V_2(\ve{x}) = \sumall \siz_i x_i$. To check the conditions of Lemma~\ref{lem:lyapunov-upper-bound}, recall that the drift of $V_2(\ve{x})$ is equal to
\begin{equation}
    \begin{aligned}
      G V_2(\ve{x}, \ve{z}) &= \sumall \arr_i \left(V_2(\ve{x}+\ve{e}_i) - V_2(\ve{x})\right) + \sumall  \sev_i z_i \left(V_2(\ve{x}-\ve{e}_i) - V_2(\ve{x})\right). 
    \end{aligned}\label{eq:proof:gv2-full}
\end{equation}
From the formula of the drift, we see that parameters $\vmx$ and $\fmx$ of Lemma~\ref{lem:lyapunov-upper-bound} satisfy $\vmx = \Lm$ and $\fmx = \sumall \arr_i \siz_i \leq \sevmax n$.

We spend the rest of the proof showing that when $V_2(\ve{x})$ is larger than a threshold, 
\begin{equation}
    \sumall \siz_i x_i = V_2(\ve{x}) > \sumall \siz_i\xnom_i + B_2, \label{eq:proof:v2-greater-B}
\end{equation}
for some non-negative $B_2$, there exists some $\gamma, \xi > 0$, such that with high probability we have $G V_2(\ve{x}, \ve{z})<-\gamma$ and in the worst case we have $G V_2(\ve{x},\ve{z}) < \xi$. We bound $GV_2$ by finding a suitable state-space concentration. To do this, we write out the form of $G V_2(\ve{x}, \ve{z})$ explicitly and bound it as below: 
\begin{align}
    G V_2(\ve{x},\ve{z}) &= \sumall (\arr_i \siz_i - \sev_i \siz_i z_i) \label{eq:proof:gv2-temp0}\\
    &= \sumall \left(\arr_i \siz_i +(\sevmax - \sev_i)\siz_i z_i\right) -  \sumall \sevmax \siz_i z_i \nonumber\\
    &\leq \sumall \left(\arr_i \siz_i +(\sevmax - \sev_i)\siz_i x_i \right) - \sevmax \min\left(\sumall \siz_i x_i, n-\Lm \right)\nonumber\\
    &= -\sevmax\min\left(\sumall \frac{\sev_i}{\sevmax} \siz_i(x_i -\xnom_i), \;(\slk - \Lm)-\sumall\left(1-\frac{\sev_i}{\sevmax}\right)\siz_i(x_i-\xnom_i)\right),\label{eq:proof:iter2:gv2-leq-sth}
\end{align}
where the inequality is due to $z_i \leq x_i$ and the $\Lm$-work-conserving condition.
The last expression suggests that in order to apply Lemma~\ref{lem:lyapunov-upper-bound}, we need to lower bound $\sumall \frac{\sev_i}{\sevmax} \siz_i(X_i -\xnom_i)$ and upper bound $\sumall(1-\frac{\sev_i}{\sevmax})\siz_i(X_i-\xnom_i)$ with high probability.

Now we construct a subset $\mathcal{E} \subseteq \R^\njt$ such that $\ve{X}\in \mathcal{E}$ with high probability, and $f(\ve{x})$ has negative drift whenever $\ve{x}\in\mathcal{E}$ and $f(\ve{x})>B_2$ for some $B_2\geq 0$. The construction of $\mathcal{E}$ is based on the high probability sets implied by Lemma~\ref{lem:lyapunov-upper-bound} and Lemma~\ref{lem:proof:workload-tail} that we just proved. For a fixed $j_1 \geq 0$, define the set
\[
    \mathcal{E}_{4}(j_1)  = \left\{\ve{x}\in\R^\njt \;\middle|\; \sumall \left(\tfrac{\sevmax}{\sev_i} + \tfrac{\sev_i}{\sevmax} - 2\right) \siz_i(x_i - \xnom_i) > - \alpha_{4} - \beta_{4} j_1\right\},
\]
where
\[
    \beta_{4} = \left(\frac{\sevmax}{\sevmin} + \frac{\sevmin}{\sevmax} - 2\right)^2\frac{\sevmax}{\sevmin} \cdot \frac{n\Lm}{\epsilon \slk}, \quad \alpha_{4} = \epsilon \slk + 2\beta_{4} \log n. 
\]
Note that this choice of $\beta_{4}$, $\alpha_{4}$ yields $\beta_{4} = \Theta\left(\frac{n\Lm}{\slk}\right)$, $\alpha_{4}= \epsilon \slk + \Theta\left(\frac{n\Lm}{\slk} \log n\right)$. In addition, let
\[
    \mathcal{E}_1(j_1) = \left\{\ve{x}\in\R^\njt \;\middle|\; \sumall \frac{\siz_i}{\sev_i} (x_i - \xnom_i) < \alpha_1 + \beta_1 j_1\right\},
\]
where $\beta_1 = \Theta\left(\frac{n\Lm}{\slk}\right)$ and $\alpha_1  = \frac{2\epsilon\slk}{\sevmin}+\Theta\left(\frac{n\Lm}{\slk}\log n\right)$ as defined in Lemma~\ref{lem:proof:workload-tail}. We define $\mathcal{E}$ and $B_2$ as 
\begin{align*}
\mathcal{E} &= \mathcal{E}_{4}(j_1)\cap \mathcal{E}_1(j_1+2\log n). \\
B_2 &= \sevmax(\alpha_1+\beta_1 (j_1+2\log n))+\alpha_{4}+\beta_{4}j_1 + \epsilon \slk\\
&= 2\left(1+\tfrac{\sevmax}{\sevmin}\right) \epsilon \slk +\Theta\left(\frac{n\Lm}{\slk}\log(n)\right) + \Theta\left(\frac{n\Lm}{\slk}\right) j_1,
\end{align*}
then $V_2(\ve{x}) > \sumall \siz_i\xnom_i + B_2$ and $\ve{x}\in \mathcal{E}$ imply the following three inequalities:
\begin{align}
    &\sumall \siz_i(x_i - \xnom_i) > \sevmax\left(\alpha_1+\beta_1 (j_1+2\log n)\right)+\alpha_{4}+\beta_{4}j_1 + \epsilon \slk, \label{eq:proof:iter2:ssc1}\\
    &\sumall \frac{\siz_i}{\sev_i} (x_i - \xnom_i) < \alpha_1 + \beta_1 (j_1+2\log n),\label{eq:proof:iter2:ssc2}\\
    &\sumall \left(\tfrac{\sevmax}{\sev_i}+\tfrac{\sev_i}{\sevmax}-2\right)\siz_i(x_i - \xnom_i) > -\alpha_{4} - \beta_{4} j_1 \label{eq:proof:iter2:ssc3}.
\end{align}
The linear combination of the three inequalities $- \eqref{eq:proof:iter2:ssc1} +  \sevmax \cdot \eqref{eq:proof:iter2:ssc2} - \eqref{eq:proof:iter2:ssc3}$ gives us 
\begin{equation}
    \sumall \left(1 - \tfrac{\sev_i}{\sevmax}\right)\siz_i (x_i - \xnom_i) \leq -\epsilon \slk < 0;
\end{equation}
another linear combination $2\cdot \eqref{eq:proof:iter2:ssc1} -  \sevmax \cdot \eqref{eq:proof:iter2:ssc2} + \eqref{eq:proof:iter2:ssc3}$ gives us
\begin{equation}
     \sumall \frac{\sev_i}{\sevmax} \siz_i(x_i - \xnom_i) \geq \sevmax\left(\alpha_1+\beta_1 (j_1+2\log n)\right)+\alpha_{4}+\beta_{4}j_1 + 2\epsilon \slk > \epsilon \slk.
\end{equation}
Note that when deriving the above two inequalities, we have used the fact that all parameters involved in the expressions, $\alpha_1, \beta_1, j_1, \alpha_{4}, \beta_{4},\slk$, are non-negative. We substitute the above two inequalities back to \eqref{eq:proof:iter2:gv2-leq-sth}, and choose $\epsilon < \min(1-\Lmratio, \Lmratio)$. This gives us
\begin{equation*}
    GV_2(\ve{x},\ve{z}) \leq -\sevmax \min(\epsilon \slk, \slk - \Lm) \leq -\sevmax \epsilon \slk.
\end{equation*}

Next, we show that $\ve{X}\in\mathcal{E}$ with high probability. We apply Lemma~\ref{lem:mminf-bounds} with $c_i = \frac{\sevmax}{\sev_i} + \frac{\sev_i}{\sevmax} - 2$, $\cmax= \frac{\sevmax}{\sevmin} + \frac{\sevmin}{\sevmax} - 2$, $\alpha = \epsilon \slk$, $\beta =  \frac{\cmax^2 \sevmax n\Lm}{\sevmin\epsilon\slk}$. It can be verified that $\alpha\beta \geq \cmax^2 \sevmax \sysvar$, so
\begin{equation}
    \Prob\left(\sumall \left(\tfrac{\sevmax}{\sev_i} + \tfrac{\sev_i}{\sevmax} - 2\right) \siz_i(X_i - \xnom_i) \leq - \alpha - \beta j\right) \leq e^{-j},
\end{equation}
Setting $j = j_1+2\log n$ and recalling the definition of $\alpha_4$, $\beta_4$, we get
\begin{equation}
    \Prob\left(\sumall \left(\tfrac{\sevmax}{\sev_i} + \tfrac{\sev_i}{\sevmax} - 2\right) \siz_i(X_i - \xnom_i) \leq - \alpha_{4} - \beta_{4} j_1\right) \leq \frac{1}{n^2} e^{-j_1},
\end{equation}
i.e., $\Prob\left(\ve{X}\in\mathcal{E}_{4}(j_1)\right) \geq 1 - \frac{1}{n^2} e^{-j_1}$. 
Moreover, applying Lemma~\ref{lem:proof:workload-tail} to the set $\mathcal{E}_1(j_1+2\log n)$, we have $\Prob\left(\ve{X} \in \mathcal{E}_1(j_1+2\log n)\right) \geq 1 - \frac{1}{n^2}e^{-j_1}$. Recall that $\mathcal{E} = \mathcal{E}_{4}(j_1)\cap \mathcal{E}_{1}(j_1+2\log n)$, so by union bound, we have
$
\Prob(\ve{X}\in\mathcal{E}) \geq 1 - \frac{2}{n^2}e^{-j_1}.
$

When $\ve{x}\notin \mathcal{E}$, from the form of \eqref{eq:proof:gv2-temp0}, we can see that $G V_2(\ve{x},\ve{z})\leq\sevmax n$.

We apply Lemma~\ref{lem:lyapunov-upper-bound} with $B = \sumall \siz_i\xnom_i + B_2$ $\gamma = -\sevmax \epsilon \slk$, $\xi = \sevmax n$, $\mathcal{E} = \mathcal{E}_{4}(j_1)\cap \mathcal{E}_{1}(j_1+2\log n)$, $\fmx \leq \sevmax n$, $\vmx = \Lm$. This gives us
\begin{equation}
    \begin{aligned}
        &\mspace{23mu} \Prob\left(\sumall \siz_i(X_i - \xnom_i) \geq B_2 + 2 \Lm \left(\left(\tfrac{\fmx}{\sevmax\epsilon\slk}+2\right)j_2 + 1\right) \right)\leq e^{-j_2}+ \left(\frac{n}{\epsilon \slk}+1\right)\frac{2}{n^2} e^{-j_1}. 
    \end{aligned} \label{eq:proof:v2-apply-lem}
\end{equation}
This inequality can be further simplified into 
\begin{equation}
    \begin{aligned}
        &\mspace{23mu} \Prob\left(\sumall \siz_i(X_i - \xnom_i) \geq B_2 + 2 \Lm \left(\frac{(2\Lmratio+1)n}{\epsilon\slk} j_2 + 1\right) \right)\leq e^{-j_2}+ e^{-j_1},
    \end{aligned}
\end{equation}
where we have used the fact that $\fmx \leq n\sevmax$, $n>\slk$, $\epsilon < \Lmratio$, and $\left(\frac{n}{\epsilon \slk}+1\right)\frac{2}{n^2} \leq 1$ which holds for $n$ large enough. Setting $j_2$ to be equal to $j_1$ and rearranging the terms, we get
\begin{equation}
    \Prob\left(\sumall \siz_i(X_i - \xnom_i) \geq \alpha_2 + \beta_2 j_1\right) \leq 2\exp(-j_1),
\end{equation}
where
\[
    \beta_2 = \frac{(2\Lmratio+1)n\Lm}{\epsilon\slk}+\sevmax \beta_1 +\beta_{4} = \Theta(\frac{n\Lm}{\slk}),
\]
\begin{align*}
    \alpha_2 &=  \sevmax(\alpha_1+2\beta_2 \log n) + \alpha_{4}+\epsilon \slk + 2\Lm\\
    &= 2(1+\frac{\sevmax}{\sevmin}) \epsilon \slk +\Theta(\frac{n\Lm}{\slk}\log(n)).
\end{align*}

We choose any $\epsilon$ satisfying $0 < \epsilon <  \min\left(1-\Lmratio, \Lmratio, \frac{1}{4(1+\sevmax/\sevmin)}\right)$ and set $j_1 \leftarrow j + \log(2)$, then Lemma~\ref{lem:proof:total-server-need-tail} is proved.

\end{proof}

\tsnexp*

\begin{proof}[Proof of Lemma~\ref{lem:proof:total-server-need-expectation}.]
We proceed by trying to build a connection between the expectation $\E\left[\sumall \siz_i (X_i - \xnom_i)\right]$ and the tail probability $\Prob\left(\sumall \siz_i (X_i - \xnom_i) \geq \alpha_2 + \beta_2 j\right)$, where $\alpha_2=\frac{\slk}{2} + \Theta\left(\frac{n\Lm}{\slk}\log n \right)$, $\beta_2=\Theta\left(\frac{n\Lm}{\slk}\right)$ are given in Lemma~\ref{lem:proof:total-server-need-tail}. 

We first decompose $\E\left[\sumall \siz_i (X_i - \xnom_i)\right]$ as follows:
\begin{align}
    \E\left[\sumall \siz_i (X_i - \xnom_i)\right]
    &= \E\left[\sumall \siz_i (X_i - \xnom_i) \indi_{\{\sumall \siz_i X_i \geq n\}}\right]+ \E\left[\sumall \siz_i (X_i - \xnom_i) \indi_{\{\sumall \siz_i X_i < n\}}\right].\nonumber\\
    &= \E\left[\left(\sumall \siz_i(X_i - \xnom_i) - \slk\right)^+ \right]\label{eq:total-server-need-expect:decomp1}\\
    &\mspace{18mu} + \slk \Prob\left(\sumall \siz_i X_i \geq n\right)\label{eq:total-server-need-expect:decomp2}\\
    &\mspace{18mu}+\E\left[\sumall \siz_i (X_i - \xnom_i) \indi_{\{\sumall \siz_i X_i < n\}}\right] \label{eq:total-server-need-expect:decomp3},
\end{align}
where in the second equality we have used the fact that $\sumall \siz_i \xnom_i + \slk = n$. 

Next, we bound the three terms separately. The first term $\E\left[\left(\sumall \siz_i(X_i - \xnom_i) - \slk\right)^+ \right]$ can be rewritten as follows:
\begin{align*}
    &\mspace{23mu} \E\left[\left(\sumall \siz_i(X_i - \xnom_i) - \slk\right)^+ \right]\\
    &= \int_0^\infty \Prob\left(\sumall \siz_i (X_i - \xnom_i) - \slk \geq t\right) dt
    \leq \sum_{j=\lfloor \frac{\slk-\alpha_2}{\beta_2}\rfloor}^\infty \int_{\alpha_2+\beta_2 j}^{\alpha_2+\beta_2(j+1)} \Prob\left(\sumall \siz_i (X_i - \xnom_i)  \geq t\right) dt\\
    &\leq \sum_{j=\lfloor \frac{\slk-\alpha_2}{\beta_2}\rfloor}^\infty \beta_2 \Prob\left(\sumall \siz_i(X_i - \xnom_i) \geq \alpha_2 + \beta_2 \cdot j\right)
    \leq \sum_{j=\lfloor \frac{\slk-\alpha_2}{\beta_2}\rfloor}^\infty \beta_2 e^{-j},
\end{align*}
where $\lfloor a \rfloor$ means the largest integer that is no greater than a real number $a$, and the last inequality is due to Lemma~\ref{lem:proof:total-server-need-tail}. Therefore,
\begin{equation}
    \E\left[\left(\sumall \siz_i(X_i - \xnom_i) - \slk\right)^+ \right] \leq \frac{e}{e-1}\beta_2 \exp\left(-\left\lfloor\frac{\slk-\alpha_2}{\beta_2}\right\rfloor\right), \label{eq:total-server-need-expect:bound-decomp1}
\end{equation}
where $e$ is the base for the natural logarithm.

Next, we try to bound the sum of the term in \eqref{eq:total-server-need-expect:decomp2} and the terms in \eqref{eq:total-server-need-expect:decomp3}, which is equal to $\Prob\left(\sumall \siz_i X_i \geq n\right) + \E\left[\sumall \siz_i (X_i - \xnom_i) \indi_{\{\sumall \siz_i X_i < n\}}\right]$. We do this by analyzing the Lyapunov drift of the function $f(\ve{x}) = \sum_{i=1}^\njt \frac{\siz_i}{\sev_i}x_i$. The drift of the sum is equal to
\begin{equation*}
    G f(\ve{x},\ve{z}) = \sum_{i=1}^\njt \left(\frac{\arr_i\siz_i}{\sev_i} - \siz_i z_i\right).
\end{equation*}
Because $f(\ve{x})$ is a linear function, and the system operates under an $\Lm$-work-conserving policy, by Lemma~\ref{lem:polynomial-zero-drift}, we have $\E\left[G f(\ve{X},\ve{Z})\right] = 0$, i.e.
\begin{equation}
    \E\left[\sum_{i=1}^\njt \left(\frac{\arr_i\siz_i}{\sev_i} - \siz_i Z_i\right)\right] = 0.
\end{equation}
Observe that by the $\Lm$-work-conserving condition, $\sumall \siz_i Z_i \geq n-\Lm$ when $\sumall \siz_i X_i \geq n$, and $\sumall \siz_i Z_i  = \sumall \siz_i X_i $ when $\sumall \siz_i X_i < n$. Therefore, 
\[
    \begin{aligned}
    \mspace{23mu}\E\left[\sum_{i=1}^\njt \left(\frac{\arr_i\siz_i}{\sev_i} - \siz_i Z_i\right)\indi_{\{\sumall \siz_i X_i \geq n\}}\right] \leq (-\slk + \Lm) \Prob\left(\sumall \siz_i X_i \geq n\right),
    \end{aligned}
\]
\[
\begin{aligned}
     \mspace{23mu}\E\left[\sum_{i=1}^\njt \left(\frac{\arr_i\siz_i}{\sev_i} - \siz_i Z_i\right)\indi_{\{\sumall \siz_i X_i < n\}}\right]\leq \E\left[\sum_{i=1}^\njt \siz_i (\xnom_i - X_i) \indi_{\{\sumall \siz_i X_i < n\}}\right].
\end{aligned}
\]
Summing up the above two equations and rearranging the terms yield
\begin{align}
   &\mspace{18mu}\E\left[\sum_{i=1}^\njt \siz_i (X_i-\xnom_i) \indi_{\{\sumall \siz_i X_i < n\}}\right] + \slk \Prob\left(\sumall \siz_i X_i \geq n\right) \nonumber \\
   &\leq \Lm \Prob\left(\sumall \siz_i X_i \geq n\right)
   = \Lm \Prob\left(\sumall \siz_i (X_i-\xnom_i) \geq \slk \right) \nonumber \\
    &\leq \Lm \exp\left(-\left\lfloor\frac{\slk-\alpha_2}{\beta_2}\right\rfloor\right) \label{eq:total-server-need-expect:bound-decomp23}.
   \end{align}
where the last inequality follows from Lemma~\ref{lem:proof:total-server-need-tail} with $j = \lfloor \frac{\slk-\alpha_2}{\beta_2}\rfloor$.

Combining \eqref{eq:total-server-need-expect:bound-decomp1} and \eqref{eq:total-server-need-expect:bound-decomp23}, we finally get
\begin{equation}
     \E\left[\sumall \siz_i (X_i - \xnom_i)\right] \leq \left(\frac{e}{e-1} \beta_2 + \Lm \right) \exp\left(-\left\lfloor\frac{\slk-\alpha_2}{\beta_2}\right\rfloor\right).
\end{equation}
Observe that
\begin{equation}
    \frac{e}{e-1}\beta_2 + \Lm = O\left(\frac{n\Lm}{\slk}\right).
\end{equation}
Because $\slk = \wbrac{\frac{n\Lm}{\slk}\log n}$, we have $\alpha_2 = \frac{\slk}{2} + o(\slk)$. Consequently,
\begin{equation}
    \left\lfloor\frac{\slk-\alpha_2}{\beta_2}\right\rfloor =\Omega\left(   \frac{\slk^2}{n\Lm}\right).
\end{equation}
Therefore, 
\[\E\left[\sumall \siz_i (X_i - \xnom_i)\right] =O\left(\frac{n\Lm}{\slk}\right)\exp\left(-\Omega\left(   \frac{\slk^2}{n\Lm}\right)\right).
\]
Because $\slk = \wbrac{\sqrt{n\Lm}\log n}$, we have $\exp\left(-\Omega\left(   \frac{\slk^2}{n\Lm}\right)\right) \leq \exp\left(-\Wbrac{(\log n)^2}\right)$, which decays faster than any polynomial. Therefore we can omit the factor $O\left(\frac{n\Lm}{\slk}\right)$.
This finishes the proof of Lemma~\ref{lem:proof:total-server-need-expectation}.
\end{proof}

\queueingprob*

\begin{proof}[Proof of Corollary~\ref{result:queueing-probability}.]
Choosing $j = \frac{\slk - \alpha_2}{\beta_2}$ in Lemma~\ref{lem:proof:total-server-need-tail}, we have
\begin{equation*}
    \Prob\left(\sumall \siz_i (X_i-\xnom_i) \geq \alpha_2 + \beta_2 \cdot j \right) 
    \leq \exp\left(-\frac{\slk-\alpha_2}{\beta_2}\right).
\end{equation*}
Recall that $\alpha_2 = \frac{\slk}{2} + \Thebrac{\frac{n\Lm}{\slk}\log n}$ and $\beta_2 = \Thebrac{\frac{n\Lm}{\slk}}$. Because we have assumed $\slk = \wbrac{\sqrt{n\Lm}\log n}$, $\alpha_2 = \frac{\slk}{2}+\obrac{\slk}$. Therefore, $\slk \geq \alpha_2 + \beta_2\cdot j$ and $\frac{\slk-\alpha_2}{\beta_2} = \Wbrac{\frac{\slk^2}{n\Lm}}$. 
It follows that
\begin{align*}
\Prob\left(\sumall \siz_i X_i \geq n\right) &= \Prob\left(\sumall \siz_i (X_i-\xnom_i) \geq \slk \right) \\
&\leq \Prob\left(\sumall \siz_i (X_i-\xnom_i) \geq \alpha_2 + \beta_2 \cdot j \right) \\
    &\leq \exp\left(-\Wbrac{\frac{\slk^2}{n\Lm}}\right).
\end{align*}
This finishes the proof of Corollary~\ref{result:queueing-probability}.

\end{proof}

\end{document}